\let\oldvec\vec% Store \vec in \oldvec
\documentclass[runningheads]{llncs}
\let\vec\oldvec% Restore \vec from \oldvec
\usepackage{graphicx}
\usepackage[utf8]{inputenc}
\usepackage[T1]{fontenc}
\usepackage{url}

\usepackage{tikz}
\usepackage{tikzscale}
\usepackage{amsfonts,amsmath,amssymb}
\usetikzlibrary{arrows, automata,calc}
\usepackage{xcolor}
\usepackage{subcaption}
\newsavebox{\tempfig}
\usepackage{thmtools}
\usepackage{thm-restate}
\usepackage{hyperref}
\usepackage{cleveref}
\usepackage{todonotes}
\usepackage{cite}
\usepackage{wrapfig,mwe}

\begin{document}
%
%\title{Creation Games on Blockchain Payment Networks}
% \title{Blockchain Payment Networks as Creation Games}
\title{Ride the Lightning: The Game Theory of Payment Channels}
% \title{The Game-Theory of Lightning}
%
\titlerunning{Ride the Lightning}
\author{Zeta Avarikioti \and
Lioba Heimbach \and Yuyi Wang \and Roger Wattenhofer}
\authorrunning{Z. Avarikioti et al.}
\institute{ETH Z{\"u}rich \\
\email{\{zetavar,hlioba,yuwang,wattenhofer\}@ethz.ch}
 }
% %
\maketitle              % typeset the header of the contribution
\begin{abstract}
Payment channels were introduced to solve various eminent cryptocurrency scalability issues. Multiple payment channels build a network on top of a blockchain, the so-called layer 2. 
In this work, we analyze payment networks through the lens of network creation games.
We identify betweenness and closeness centrality as central concepts regarding payment networks. 
We study the topologies that emerge when players act selfishly and determine the parameter space in which they constitute a Nash equilibrium. Moreover, we determine the social optima depending on the correlation of betweenness and closeness centrality. When possible, we bound the price of anarchy. We also briefly discuss the price of stability.

% Payment channels emerged out of the scalability issues typically faced by cryptocurrencies. Together payment channels can build a network -  otherwise referred to as Layer 2 - on top of the blockchain. Possibly being part of the future of cryptocurrencies, the study of payment networks has become increasingly important. In this work, we study payment networks in a game theoretic setting.
% We define a network creation game to model these networks, having players in the game pursue both betweenness and closeness centralities. With these incentives, players selfishly find their best strategy given a network configuration; finding the best strategy is NP-hard as we show. With our knowledge about the social optimum of our game and Nash equilibria attained in our parameters space, we bound the price of anarchy. 

\end{abstract}
\section{Introduction} 

\subsection{Motivation}
Bitcoin~\cite{nakamoto2008bitcoin} and other cryptocurrencies~\cite{wood2014ethereum,hopwood2016zcash,vansaberhagen2013cryptonote} are electrifying the world. Thanks to a distributed data structure known as the blockchain, cryptocurrencies can execute financial transactions without a trusted central authority. However, every computer participating in a blockchain must exchange, store and verify each and every transaction, and as such the transaction throughput of blockchains is embarrassingly low. The Bitcoin blockchain for instance does not process more than seven transactions per second.

With seven transactions per second, Bitcoin cannot rival established payment systems such as Visa, WeChatPay, or PayPal. Consequently, various research groups have proposed a blockchain paradigm shift
%proposed a revolution on top of the revolution :-(
-- \textit{payment channels} \cite{spilman2013channels,DW2015channels, poon2015lightning}. All payment channels follow the same basic principle: Instead of sending every transaction to the blockchain, transactions are only exchanged between the involved parties. If Alice and Bob expect to exchange multiple payments, they can establish a payment channel. The channel is set up with a blockchain funding transaction. Once the channel is available, Alice and Bob exchange all payments directly, by sending each other digitally signed payment messages. If Bob tries to cheat Alice, Alice can show the signed payment messages as a proof to the blockchain, using the original funding transaction as security.

Instead of establishing a payment channel to every other person and company in the world, thanks to a technique called Hash Time Locked Contracts (HTLCs)~\cite{bitcoinwiki:htlcs, DW2015channels,poon2015lightning}, payments can also be sent atomically through a path of payment channels. 
More precisely, each payment channel is now an edge in a \textit{payment network}, and payments will be routed along a path of payment channels in the payment network. Such a payment network is called the layer 2 of the blockchain, the blockchain itself being the layer 1.

The payment channels/networks have many significant advantages over vanilla blockchains: %that they are going to reshape cryptocurrencies and other blockchain applications. 
With payment channels, the transaction throughput becomes unlimited, as each transaction is only seen by the nodes on the path between sender and receiver of a payment. This is  like sending a packet in the internet instead of sending every packet to a central server. Solving the throughput problem will also drastically decrease transaction fees. In addition, payments will be instantaneous, as one does not have to wait multiple minutes before the blockchain verifies a transaction. Payment networks also allow for more privacy as transactions are only seen by the parties involved. On the negative side, to set up a channel, the channel owner(s) must lock some capital. However, whenever a payment channel routes a transaction on behalf of other parties, the channel owner(s) can collect a transaction fee.

Payment networks are currently a hot topic in blockchain research. In practice, the first payment networks have been deployed, and are being actively used. Prominent examples are Bitcoin's Lightning network~\cite{poon2015lightning,deckereltoo} with more than 30,000 active channels, or Ethereum's Raiden network~\cite{raiden2017}.

As Bitcoin's Lightning network is growing quickly, we need to understand these newly forming payment networks.    
Which channels will be created, and what will the network topology eventually look like? 
Network creation games~\cite{fabrikant2003on} are a perfect tool to understand these questions, since they capture the degradation of the network's efficiency when participants act selfishly.  
% Hence, we employ network creation games to model the decentralized operation of layer 2 networks.
    
    In a network creation game, the incentive of a player is to minimize her cost by choosing to whom she connects. 
    In our model, players weigh the benefits they receive from using payment channels against the channels' creation cost, and selfishly initiate connections to minimize their cost. 
    There are two types of benefits for each player: (i) the forwarding fees she receives for the transactions she routed through her channels, (ii) the reduced cost for routing her transactions through the payment network in comparison to publishing the transactions on the blockchain (blockchain fee). On the other hand, establishing a channel costs the blockchain fee. Thus, a player has to balance all these factors to decide which channels to establish to minimize her cost.
    Our goal is to gain a meaningful insight on the network structures that will emerge and evaluate their efficiency, in comparison to centralized structures designed by a central authority that previous work has shown to be almost optimal.
    
    \subsection{Our Contributions}
    In this work, we provide a game-theoretic approach to analyze the creation of blockchain payment networks.
    Specifically, we adopt betweenness centrality, a natural measure for fees a player is expected to receive by forwarding others' transactions on a path of payment channels. On the other hand, we employ closeness centrality as an intuitive proxy for the transaction fees encountered when executing transactions through other players in the network. We reflect the cost of payment channel creation by associating a price with link creation.
    Therefore, we also generalize previous work on network creation games as our model combines both betweenness and closeness centralities.
    
    Under this model, we study the topologies that emerge when players act selfishly. A specific network structure is considered a Nash equilibrium when no player can decrease her cost by unilaterally changing her connections. We examine various such structures and determine the parameter space in which they constitute a Nash equilibrium. Moreover, we determine the social optima depending on the correlation of betweenness and closeness centrality. When possible, we bound the price of anarchy, the ratio of the social costs of the worst Nash equilibrium and the social optimum~\cite{koutsoupias1999worst}, to obtain insight into the effects of lack of coordination in payment networks when players act selfishly. Furthermore, we briefly discuss the price of stability, the ratio of the social costs of the best Nash equilibrium and the social optimum~\cite{anshelevich2008price}, specifically concerning the parameter values that most accurately represent blockchain payment networks.
        
    The omitted proofs can be found in Appendix \ref{app:proofs}.
     
\subsection{Related Work}
Various payment channel protocols have been proposed in literature~\cite{poon2015lightning,DW2015channels,green2017bolt,spilman2013channels,avarikioti2019brick,Miller2017sprites,lind2019teechain,avarikioti2020cerberus}, all presenting different solutions on how to create payment channels. However, our work is independent of the channel construction specifications and thus applies to all such solutions.
        
Payment networks have been studied from an algorithmic (not game theoretic) viewpoint by Avarikioti et al.~\cite{Avarikioti2018payment,Avarikioti2018algorithmic}. In \cite{Avarikioti2018payment}, they examined the optimal graph structure and fee assignment to maximize the profit of a central authority that creates the payment network and bears the relevant costs and benefits. Furthermore, in \cite{Avarikioti2018algorithmic}, they investigated the online and offline computation of a capital-efficient payment network for a central authority. In contrast, our work studies the decentralized payment network design, where the network is created by multiple participants and not a single authority. This model reflects more accurately the currently operating payment networks, which are indeed created by thousands of users rather than a single company, following the decentralized philosophy of cryptocurrencies like Bitcoin.
        
Network creation games were originally introduced by Fabrikant et al.~\cite{fabrikant2003on}.
In their game, referred to as sum network creation game, a player unilaterally creates links to minimize the sum of distances to other players in the network (closeness centrality). Later, Albers et al.~\cite{albers2014nash} improved the upper bound for the price of anarchy and also examined a weighted network creation game. While these works solely focus on a player's closeness centrality, our model is more complex and additionally includes another metric, the players' betweenness centrality that represents the importance of a player in the network.
      
In parallel, network creation games were expanded to various settings. The idea of bilateral link creation was introduced by Corbo and Parkes~\cite{corbo2005price}. Demaine et al.~\cite{demaine2012price} devise the max game, where players try to minimize their maximum distance to any other player in the game. Intrinsic properties of peer-to-peer networks are taken into account in the network creation variation conceived by Moscibroda et al.~\cite{moscibroda2006topologies,moscibroda2011topological}. The idea of bounded budget network creation games was proposed by Ehsani et al.~\cite{ehsani2015bounded}. In bounded budget network creation games, players have a fixed budget to establish links. Nodes strive to minimize their stretch, the ratio between the distance of two nodes in a graph, and their direct distance. Moreover, Àlvarez et al.~\cite{alvarez2016celebrity} introduced the celebrity game, where players try to keep influential nodes within a fixed distance. However, the objectives in all these games give little insight to the control a player has over a network. This control is desired by players in blockchain payment networks to maximize the fees received for routing transactions, in essence their betweenness centrality.
       
A bounded budget betweenness centrality game was introduced by Bei et al.~\cite{bei2011bounded}. Given a budget to create links, players attempt to maximize their betweenness centrality. Due to their complexity, betweenness network creation games yield limited theoretical results, in comparison to those of the sum network creation game, for instance. In contrast to our work, a players closeness centrality is not taken into account in \cite{bei2011bounded}. Thus, this model is insufficient for our purpose since it does not consider how strategically connected is a player that wants to route many transactions through the payment network.
 
Buechel and Buskens \cite{buechel2013dynamics} compare betweenness and closeness centralities; however, not in a network creation game setting, as their notion of stability does not lead to Nash equilibria. We, on the other hand, study the combination of betweenness and closeness incentives in a network creation game setting.

\section{Preliminaries and Model}
 
In this section, we first introduce the essential background and assumptions for our payment network creation game, and then we introduce the necessary notation and the game-theoretic model.
% In this section, first we introduce the essential background and assumptions for our payment network creation game.  Then, we introduce the necessary notation and the game-theoretic model.
    
\subsection{Payment Networks}
        
Payment channels operate on top of the blockchain (Layer 2) and allow instant off-chain transactions. Generally, a channel is set up by two parties that deposit capital in a joint account on the blockchain. The channel can then be used to make arbitrarily many transactions without committing each transaction to the blockchain. When opening a channel, the parties pay a blockchain fee and place capital in the channel. The blockchain fee is the transaction fee to the miner, paid to have the transaction included in a block and thereby published on the blockchain.  The deposited capital funds future channel transactions and is not available for other transactions on the blockchain during the channel's lifetime. 

In our model, we assume a player single-handedly initiates a channel to a subset of other players. Incoming channels are always accepted and once installed, the channels can be used to send money in both directions (from sender to receiver, and vice versa). While any player can typically choose the amount to lock in a channel, we assume that the locked capital placed in all channels is high enough to be modeled as unlimited. In particular, we assume that all players are major (large companies, financial institutions etc.) that have thus access to large amounts of temporary capital. It is natural to assume only major players to participate in the network creation game. Typically, a market is created when there is demand for a service. Thus eventually, the payment network will be dominated by service providers that will individually connect with clients and act as intermediaries for all transactions. In this work, we only consider the flow of transactions through these service providers.
Therefore, the cost of opening a channel in our model solely reflects the permanent cost, i.e.\ the blockchain fee, and is set to 1 (wlog). Furthermore, since we assume major players only, the transactions between the players can be considered uniform.
         
In addition to enabling parties connected by a channel to exchange funds off-chain, payment channels can also be used to route off-chain transactions between a sender and receiver pair not directly connected by a channel. Transactions between the sender and receiver can be routed securely through a path of channels. Since we assume that all channels are funded with unlimited capital, the channel funds cannot deplete, and so any path in the payment network between sender and receiver is viable.
        
Together, the payment channels form a payment network. In the network, players receive a payment when transactions are routed through their channels. This payment is a transaction fee, which is typically proportional to the value of the routed transaction, to compensate the intermediate node for the loss of her channel's capital capacity. However, we consider a fixed fee for all nodes, independent of the routed value, since we assume unlimited channel capital.

\subsection{Formal Model}

A payment network can be formally expressed by an unweighted undirected graph consisting of $V$ nodes, representing the set of players, and $E$ edges, representing the set of payment channel between the players.

A payment network game consists of $n$ players $V = \{0, 1, \dots , n-1\}$, denoted by $[n]$. The strategy of player $u$ expresses the channels she chooses to open and is denoted by $s_u$, and the set $S_u = 2^{[n]-\{u\} }$ defines $u$'s strategy space. 
We denote by $G[s]$ the underlying undirected graph of $G_0[s] = \left([n], \bigcup  _{u \in [n]} \{u\} \times s_u\right)$, where $s =(s_0, \dots , s_{n-1}) \in S_0 \times \dots \times S_{n-1}$ is a strategy combination. 
Note that while a channel can possibly be created by both endpoints, this will never be the case in a Nash equilibrium.

% \vspace{-11pt}
\subsubsection{Betweenness centrality.}
The fees received by a player for providing gateway services to other players' transactions  are modeled by her betweenness centrality. Betweenness centrality was first introduced as a measure of a player's importance in a social network by Freeman et al.~\cite{freeman1978centrality}. According to \cite{freeman1978centrality}, the betweenness centrality of a player $u$ in a graph $G(V,E)$ is %given by 
$\sum \limits _{\substack{s,r \in V \\  s \neq r\neq u, m(s,r) > 0}} \dfrac{m_u (s,r)}{m (s,r)},$
where $m_u (s,r)$ is the number of shortest paths between sender $s$ and receiver $r$ that route through player $u$ and $m(s,r)$ is the total number of shortest paths between $s$ and $r$. Additionally, $s \neq r\neq u$ indicates that $s \neq r$, $s \neq u$ and $r \neq u$. Intuitively, the betweenness centrality of player $u$ is a measure of the expected number of sender and receiver pairs that would choose to route their transactions through her in a payment network. Providing an insight into the transaction fees a player is expected to receive, the betweenness centrality lends itself to reflect the motivation of a player in a payment network to maximize the payments secured through providing transaction gateway services.

However, in our model, the betweenness of player $u$ is measured as follows:
$$\text{betweenness}_u(s) = (n-1)(n-2) - \sum \limits _{\substack{s,r \in [n]: \\  s \neq r\neq u, m(s,r) > 0}} \dfrac{m_u (s,r)}{m (s,r)}.$$ 
We subtract $u$'s betweenness centrality, as defined by Freeman et al.~\cite{freeman1978centrality}, from her maximum possible betweenness centrality to ensure that the social cost is always positive - avoiding cases where price of anarchy is undefined.

% \vspace{-11pt}
\subsubsection{Closeness centrality.}
Furthermore, we model the fees encountered by a player when having her transactions routed through the network with her closeness centrality. Closeness centrality measures the sum of distances of player $u$ to all other players and is given by
        $$\text{closeness}_u(s) = \sum \limits _{r \in [n]-u} \left(d_{G[s]}(u,r) -1\right),$$
        for a player $u$, where $d_{G[s]}(u,r)$ is the distance between $u$ and $r$ in the graph $G[s]$. With the transaction fees fixed per edge in our model, the distance to a player $r$ estimates the costs encountered by player $u$ when sending a transaction to player $r$. Therefore, the sum of distances to all other players is a natural proxy for the fees $u$ faces for making transactions when assuming uniform transactions. 

Thus, the combination of betweenness and closeness centralities accurately encapsulates the incentives inherent to players in a blockchain payment network. 

% \vspace{-11pt}
\subsubsection{Cost.} The cost of player $u$ under the strategy combination $s$ is
   $\text{cost}_u(s) = \lvert s_u \rvert + b \cdot \text{betweenness}_u(s) + c \cdot  \text{closeness}_u(s),$
    where $b \geq 0$ is the betweenness weight and $c >0$ the closeness weight.
    Letting $c>0$ ensures that the graph is always connected, as a player's cost is infinite in a disconnected graph.   
    Additionally, the model assumes the same price for all nodes and embeds this into coefficients $b$ and $c$. While this does not exactly encapsulate reality, it is a reasonable assumption. Different paths offer similar services to payers. In such a setting, Bertrand competition~\cite{bertrand1883book} suggests that competition will drive the prices from different players to be within a close region of each other. 

% \vspace{-11pt} 
\subsubsection{Social optimum.}  The objective of player $u$ is to minimize her cost,
    $\min_{s_u} \text{cost}_u(s)$. 
    The social cost is the sum off all players' costs,
    $\text{cost}(s) = \sum _{u \in [n]} \text{cost}_u(s)$. 
    Thus, the social optimum is  $\min  _{s} \text{cost}(s)$.

\section{Payment Network Creation Game}
    To gain an insight into the efficiency of emerging topologies when players act selfishly, we will first analyze the social optimum for our model. After studying if and when prominent graphs are Nash equilibria, we conclude by bounding the price of anarchy and  the price of stability.
    
    \subsection{Social Optimum}
        By the definition of the cost function, the social cost is
        $$\text{cost}(s) =\sum \limits _{u \in [n]} \text{cost}_u(s)= \lvert E(G) \rvert + b  \sum \limits _{u \in [n]} \text{betweenness}_u(s) + c \sum \limits _{u \in [n]} \text{closeness}_u(s),$$
        for any graph where no channel is paid by both endpoints. This constraint is met for all Nash equilibria. To lower bound the social cost, we will first simplify the social cost expression. Lemma~\ref{lem:secondlemma} shows how to express the social cost directly in terms of the number of edges and the sum of the players' closeness centrality costs, facilitating further analysis.
        
        \begin{restatable}[]{lemma}{socialcost}\label{lem:secondlemma}
            The social cost in $G$ is given by 
            $ \text{cost}(s) = \lvert E(G) \rvert  + b \cdot n\cdot (n-1)(n-2)+ (c- b) \cdot \sum \limits _{u \in [n]} \text{closeness}_u(s)$.
        \end{restatable}
        
    The distance of a vertex $v$ of a connected graph $G$ is $d(v) := \sum_{u\in [n]-v} d_G(v,u)$. The distance of a connected graph $G$ is $d(G) := \sum_{v\in [n]} d(v)/2$. If $G$ is not connected, then $d(v) = \infty$ for any $v$, and $d(G) = \infty$. 
        
        \begin{lemma}[Theorem 2.3~\cite{entringer1976distance}]\label{lem:thirdlemma}
            If  $G$ is a connected graph with $n$ vertices and $k$ edges then $n\cdot (n - 1)  \leq   d(G)  +  k \leq   \frac{1}{6}\cdot \left( n^3 - 5 \cdot n -6 \right)$.
        \end{lemma}
        Lemma~\ref{lem:thirdlemma} provides bounds for the distance of a graph $G$, 
        $$d(G) = \dfrac{1}{2} \sum  _{u \in [n]}  \sum  _{r \in [n]-u} d_{G}(u,r)$$ which is
        useful for finding the social optimum for our game.
        % In \cite{entringer1976distance} Lemma~\ref{lem:thirdlemma} is proven and the path graph achieves the upper bound. This can be used to find the social optimum. Dependent on the weights $b$ and $c$, the social optimum for our payment network creation game is given in Theorem~\ref{thm:firsttheorem} and Figure~\ref{fig:socialoptimum} illustrates this in the parameter space.
        
        % \begin{figure}[hbt!]
        %     \centering
        %     \includegraphics[width = 0.6\textwidth]{socialOptimum.png}
        %     \caption{Parameter map for social optimum of game.}
        %     \label{fig:socialoptimum}
        % \end{figure}
        % \vspace{-11pt}
        \begin{figure}[!]
           
            \begin{wrapfigure}{r}{0.5\textwidth} 
            \vspace{-30pt}
            \centering
            \includegraphics[width = 0.4\textwidth]{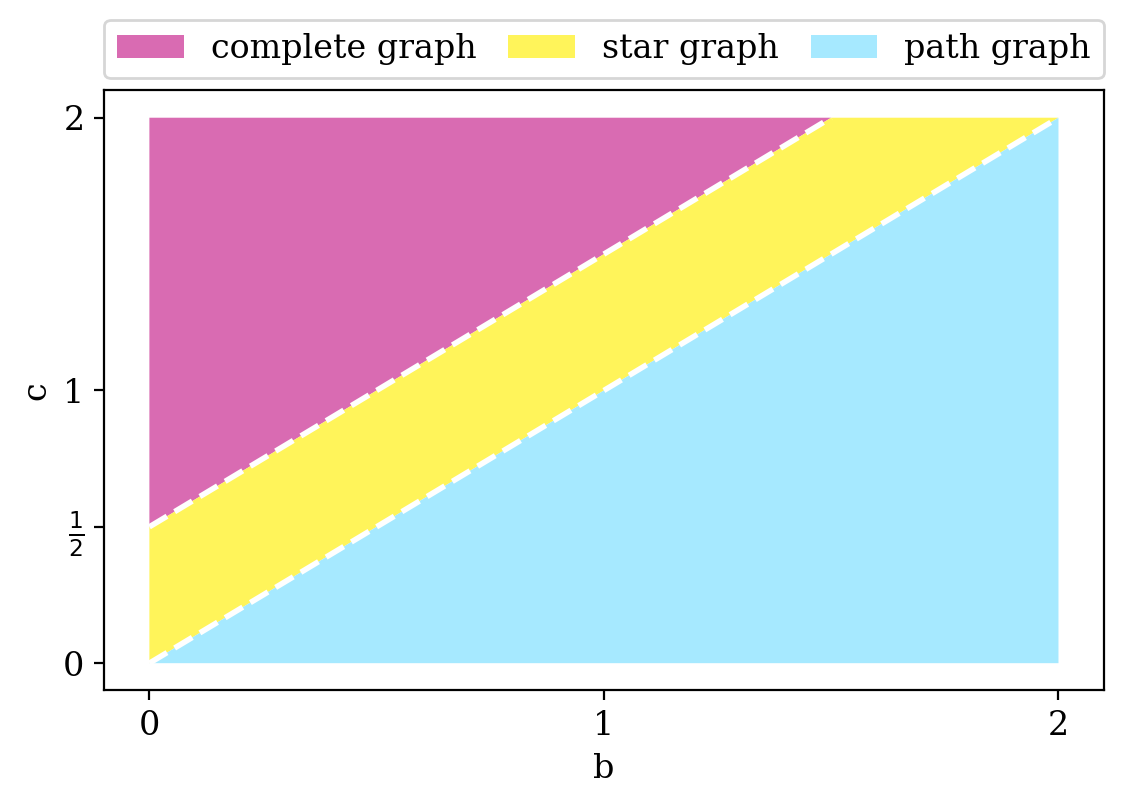}
            \caption{\small Parameter map for social optimum.}
            \label{fig:socialoptimum}
            %\vspace{0pt}
        \end{wrapfigure}
        
        \normalsize
        
        In \cite{entringer1976distance} Lemma~\ref{lem:thirdlemma} is proven and stated that the path graph achieves the upper bound; maximizes the distance term. This can be used to find the social optimum. Dependent on the weights $b$ and $c$, the social optimum for our payment network creation game is given in Theorem~\ref{thm:firsttheorem}, and illustrated in Figure~\ref{fig:socialoptimum}.
        \end{figure}
        
        % \vspace{-11pt}
        \begin{restatable}{theorem}{socialopt}\label{thm:firsttheorem}
            The social optimum is a complete graph for $c > \frac{1}{2} + b$, a star graph for $b \leq c \leq \frac{1}{2} + b$ and a path graph for $c < b$. 
        \end{restatable}

        In areas most accurately describing payment networks, we expect the weights $b$ and $c$ to be smaller than the cost of channel creation and close to each other. For these cases, we observe the star graph is the social optimum.
        % Thus, the star graph is the social optimum in areas of our parameter space we would expect to describe payment networks most accurately.  
        
    \subsection{Nash Equilibria}\label{chap:Nash} 
        To find a Nash equilibrium, one could follow a naive approach: start with a fixed graph structure and then continuously compute a player's best response in the game.     
        % A naive approach to finding Nash equilibria is to start with a graph and continuously compute a player's best response in the game. Adjusting the strategy according to the best response until a Nash equilibrium is reached. 
        However, Theorem~\ref{thm:secondtheorem}  shows that it is NP-hard to calculate a player's best response. 
        \begin{restatable}[]{theorem}{np}\label{thm:secondtheorem}
            Given a strategy $s \in S_0 \times \dots  \times S_{n-1}$ and $u \in [n]$, it is NP-hard to computed the best response of $u$.
        \end{restatable}
        
        Therefore, with this in mind, we analyze prominent graph topologies theoretically, to see if and when they are Nash equilibria in our game. The results are illustrated in Figure~\ref{fig:parametermaps}. However, complementary to the theoretical analysis we also run a simulation to get insights into emerging graph topologies for a small number of players.
        
        \begin{figure}[hbt!]
            \begin{subfigure}[b]{0.32\textwidth}
                \savebox{\tempfig}{\includegraphics[width = \linewidth]{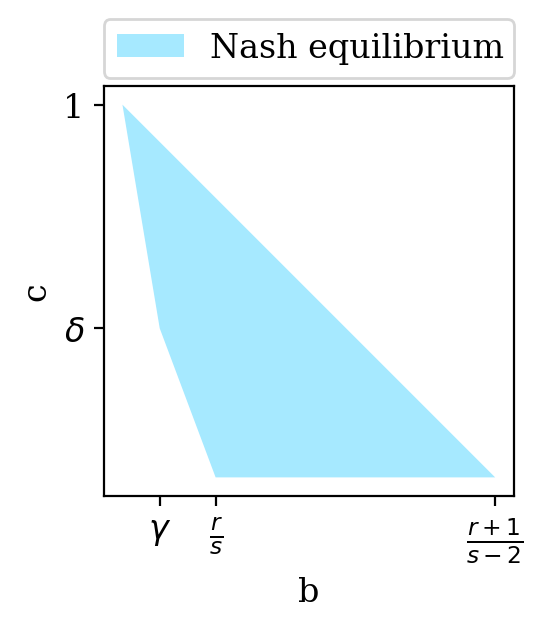}}% Store larger image in box
                \raisebox{\dimexpr\ht\tempfig-\height}{\includegraphics[width = \linewidth]{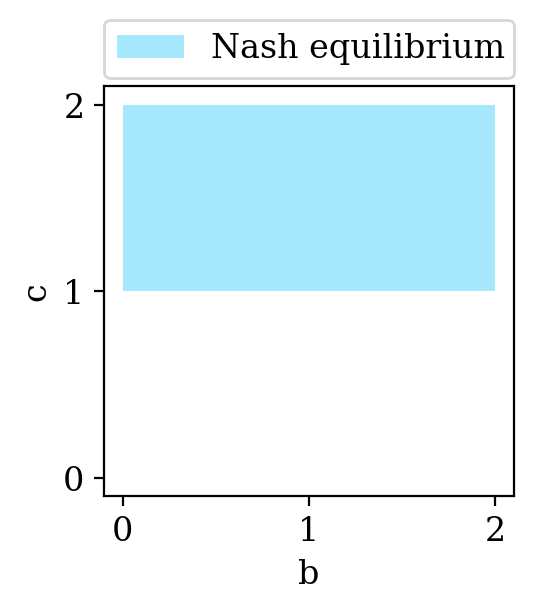}}
               
                \caption{complete graph\label{fig:completegraph}}
            \end{subfigure}
            \hfill
            \begin{subfigure}[b]{0.32\textwidth}  
                \includegraphics[width = \linewidth]{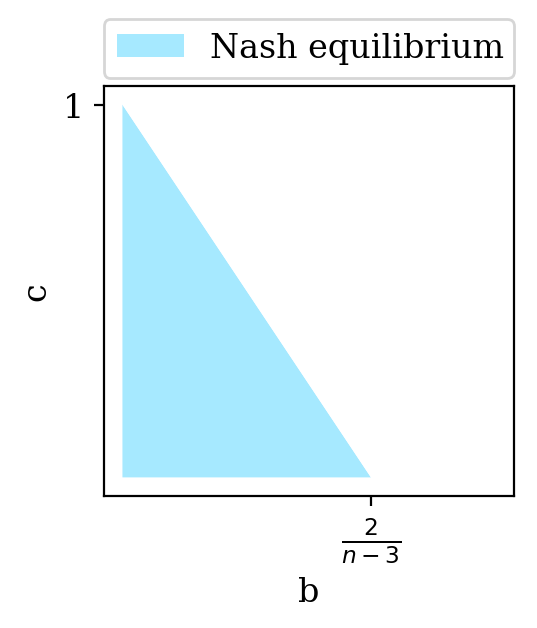}
                \caption{star graph ($n \geq 4$)\label{fig:stargraph}}
            \end{subfigure}
            \hfill
            \begin{subfigure}[b]{0.32\textwidth}
                \includegraphics[width = \linewidth]{completeBipartiteGraph.png}
                \caption{biclique \label{fig:biclique}}
            \end{subfigure}
            \caption{\small Parameter map for prominent graphs. In Figure \ref{fig:biclique}, $r$ and $s$ are the subset sizes ($3 \leq r \leq s$). With $\alpha  = \frac{s \cdot (s-1)}{r \cdot (s-2)}$ and $\beta = \frac{1}{s-r+1} \left( \frac{s\cdot (s-1)}{r} - \frac{(r-2)(r-1)}{s+1} \right)$, $(\gamma, \delta )$ is the intersection between  $1 =  \frac{s}{r} b + \frac{s+r-3}{s-1} c$ and $1 = \min \left\{\alpha , \beta \right\} \cdot  b + c $.}
            \label{fig:parametermaps}
        \end{figure} 
        
        % \vspace{-11pt}
        \subsubsection*{Complete Graph.}
            For large values of $c$ the complete graph is the only Nash equilibrium as stated in Theorem~\ref{thm:thridtheorem}.
            Additionally, the complete graph is also a Nash equilibrium for $c =1$, but it is not necessarily the only one. 
            However, for small values of $c$, which are the values we expect to encounter in a payment network creation game, the complete graph is not a Nash equilibrium, as stated in Theorem~\ref{thm:fourththeorem}.

            \begin{restatable}{theorem}{cliqueone}\label{thm:thridtheorem}
                For $c>1$, the only Nash equilibrium is the complete graph. 
            \end{restatable}
            % \begin{proof}
            %     The addition of an edge by a player never increases her betweenness cost. Thus, by the definition of the cost function any Nash equilibrium cannot be missing any edges whose addition would reduce a players closeness by more than 1, the cost of building an edge. As $c>1$, no edge can be missing in the graph and the only Nash equilibrium is the complete graph.\hfill \qed 
            % \end{proof}

            \begin{restatable}{theorem}{cliquetwo}\label{thm:fourththeorem}
                For $c<1$ and $n \geq 3$, the complete graph is never a  Nash equilibrium. 
            \end{restatable}
            % \begin{proof}
            %     In a complete graph the removal of an edge by a player does not change her betweenness cost and her closeness cost is increased by $c$. Thus, the cost of a player would decrease when removing one edge. Therefore, the complete graph is not a Nash equilibrium for $c<1$.\hfill \qed 
            % \end{proof}

           Figure~\ref{fig:completegraph} visualizes the combination of these results, i.e., when the complete graph is a Nash equilibrium in our game. 
           We observe that for some weight combinations the complete graph is both the social optimum and a Nash equilibrium. However, most payment networks are not expected to fall into this area of the parameter space.
            
        % \vspace{-11pt}
        \subsubsection*{Path Graph.}\label{sec:path}
            While the path graph is the social optimum for a significant area of the parameter space, we show it can only be a Nash equilibrium for small sets of players. For $n =3$, the path graph is a Nash equilibrium for all $c\leq 1$, as it is the only possible connected graph that is not the complete graph.  
            
            \begin{restatable}[]{proposition}{patha}\label{pro:fristpathproposition}
                For $n=4$, the path graph is a Nash equilibrium if and only if $ 1 \leq b +2 \cdot c$. 
            \end{restatable}
            
            % \vspace{-11pt}
            \begin{restatable}[]{proposition}{pathb}\label{pro:secondpathproposition}
                For $n=5$, the path graph is a Nash equilibrium if and only if $ 1 \leq 2\cdot b +4\cdot c  $.
            \end{restatable}
            
            Propositions~\ref{pro:fristpathproposition} and~\ref{pro:secondpathproposition} identify when the path graph is a Nash equilibrium for networks with four and five players respectively. These bounds partly overlap with areas in which the path graph is the social optimum. While this partial correspondence between the Nash equilibrium and social optimum appears promising for the coordination of our game, Theorem~\ref{thm:paththeorem} suggests to the contrary.
            
            \begin{restatable}{theorem}{paththeorem}\label{thm:paththeorem}
                For $n \geq 6$, the path graph is never a Nash equilibrium.
            \end{restatable}

            Hence, the path graph is not expected to be a Nash equilibrium for payment networks that typically consist of many nodes.
            % For networks in the dimension of payment networks the path graph is therefore never a Nash equilibrium. However, this does not preclude networks with a similar cost being Nash equilibria. 
        
        % \vspace{-11pt}    
        \subsubsection*{Circle Graph.}
            The results we find for the circle graph are similar to those for the path graph. For small values of $n$, the  circle graph can be a Nash equilibrium depending on the weights $b$ and $c$. The circle graph and the complete graph are the same for $n = 3$. Thus, for $n=3$ the circle graph is a Nash equilibrium if and only if $c \geq 1$.
            
            \begin{restatable}[]{proposition}{stara}\label{pro:fristproposition}
                For $n=4$, the circle graph is a Nash equilibrium if and only if $c \leq 1 \leq b +2 \cdot c$. 
            \end{restatable}
            
            \vspace{-11pt}
            \begin{restatable}[]{proposition}{starb}\label{pro:secondproposition}
                For $n=5$, the circle graph is a Nash equilibrium if and only if $b + c \leq 1 \leq 2\cdot b +4\cdot c  $.
            \end{restatable}
            
            Propositions~\ref{pro:fristproposition} and~\ref{pro:secondproposition} show that for small $n$, the circle graph can be a Nash equilibrium depending on the weights $b$ and $c$. However, for large $n$ the circle graph is never a Nash equilibrium, as stated in Theorem~\ref{thm:fifththeorem}. 
                
            \begin{restatable}{theorem}{circle}\label{thm:fifththeorem}
                There exists a $N > 0$, such that for all $n \geq N$ the circle graph is never a Nash equilibrium. 
            \end{restatable}
            We note that simulations suggest that for $n \geq 6$ the circle graph is never a Nash equilibrium. Parameter sweeps indicating that $N= 6$ can be found in Appendix~\ref{app:simulationcircle}.
        
        % \vspace{-11pt}    
        \subsubsection*{Star Graph.}\label{sec:star}
            The star graph is the social optimum for a significant part of our parameter space. In a star graph the player in the center has minimal closeness and betweenness costs; all other players have maximal betweenness cost. While this does not directly appear to be a stable network, Theorem~\ref{thm:star} suggests that the star graph is a Nash equilibrium for smaller values of $b$ and $c$. These results are depicted in Figure~\ref{fig:stargraph}.

            \begin{theorem}\label{thm:star}
                For $n \geq 4$, the star graph is always a Nash equilibrium if and only if $0 \leq 1 -  \frac{n-3}{2} b -c$.
            \end{theorem}
            \begin{proof} 
                To show that the star is always a Nash equilibrium for $n \geq 4$ and $0 \leq 1-\frac{n-3}{2} b-c $, we will consider a star graph consisting of $n$ players $V = \{0,1, \dots , n-1\}$. Without loss of generality we assume that player $0$ is the center of the star. 
                
                No player in the star graph has an incentive to remove an edge, as this would lead to infinite cost. Thus, player $0$ has no incentive to change strategy, as she is connected to everyone. 
                
                Next we consider star graphs where all links are initiated by player $0$ and star graphs where at least one link is initiated by another player separately.
            	
            	If all links are initiated by player $0$, players $1$, $2$, \dots, $n-1$ are all in an equivalent position and it is therefore sufficient to solely consider player $1$. Player $1$ would only add links, if this leads to a decrease in her cost. Initiating an edge to player $0$ would only increase her cost. Additionally, for the remaining $n-2$ players, it only matters to how many player $1$ connects. The change in cost when adding $m$, where $1 \leq m \leq n-2$, edges is given by 
                % \begin{align*}
                    $\Delta \text{cost}_1(\text{add $m$ links}) =  m  -   \dfrac{m \cdot (m-1)}{2} b - m\cdot c .$
                % \end{align*}
                Thus, player $1$ will change strategy if $\Delta \text{cost}_1(\text{add $m$ links})< 0$. The change in cost is minimized for $m= n-2$. 
                
                In star graphs where at least one player other than $0$ initiates a link, players that have no outgoing links are in the same position as those analyzed previously. Thus, it suffices to  consider player $i$, where $i \neq 0$, that has one outgoing link. In addition to only initiating new links, player $i$ can remove the link to player $0$ and initiates $l$, where $1 \leq l \leq n-2$, new links. The change in cost is then given as
                % \begin{align*}
                    $\Delta \text{cost}_i(\text{add $l$ links}) =  (l-1)  -   \dfrac{l \cdot (l-1)}{2} b - (l-1)\cdot c .$
                % \end{align*}
                However, this leads to more restrictive bounds and there is no need for players other than player $0$ to have outgoing links. 
                
                Thus, the star is a Nash equilibrium if and only if $0 \leq 1 - \frac{n-3}{2}b- c.$ \hfill \qed 
            \end{proof}
            
            We note that the areas where the star is both a Nash equilibrium and the social optimum overlap partially.
        
        % \vspace{-11pt}    
        \subsubsection*{Complete Bipartite Graph.}
            The star graph is a complete bipartite graph where one group has size one. In this section, we analyze more general complete bipartite graphs or bicliques $K_{r,s}$, where $r$ is the size of the smaller subset and $s$ is the size of the larger subset. In a complete bipartite graph, every node from one subset is connected to all nodes from the other subset. 
            \begin{theorem}
               The complete bipartite graph $K_{r,s}$ with $3 \leq r \leq s$ is stable if and only if $ \frac{s-2}{r+1} b + c   \leq 1 \leq \min \left\{ \frac{s}{r} b + \frac{s+r-3}{s-1} c, \min \left\{\alpha , \beta \right\} \cdot  b + c\right\} $, where $\alpha  = \frac{s \cdot (s-1)}{r \cdot (s-2)}$ and $\beta = \frac{1}{s-r+1} \left( \frac{s\cdot (s-1)}{r} - \frac{(r-2)(r-1)}{s+1} \right)$. 
            \end{theorem}
            \begin{proof}
                Additional links can only be created within a subset in a complete bipartite graph. Similarly to adding links in a star graph, the change in cost when adding $m$ links is given by
                % \begin{align*}
                    $\Delta \text{cost}_u(\text{add $m$ links}) = m - \dfrac{m \cdot (m-1)}{l+1} b- m \cdot c,$
                % \end{align*}
                where $l\in \{r,s\}$ is the size of the subset not including the player. 
                
                A player changes strategy when $\Delta \text{cost}_u(\text{add $m$ links})< 0$. The change in cost is minimized when $m$ is maximized and $l= r$. $m$ can therefore be $s-1$ at most. Thus, the upper bound for $K_{r,s}$ being a Nash equilibrium is  
                $1 \geq \dfrac{s-2}{r+1} b+c.$
            	
            	Players in the subset of size $r$, benefit more from a link to the other subset, as their betweenness cost is smaller. Thus, players from the larger subset with outgoing links would change strategy sooner. In the case where the subsets are of equal size, the link direction does not matter. Hence, to find a lower bound for $b$ and $c$ we only consider complete bipartite graphs, in which all links are established from the smaller subset, as seen in Figure~\ref{fig:bip1}. If players from the larger subset woul   Without loss of generality we will only consider player $u$ in the following analysis. It is not reasonable for player $u$ to remove all her links without adding any new links, as her cost would become infinite. Depending on the other parameters, it might be more optimal to remove all her previous links and only connect to one player in her subset (Figure~\ref{fig:bip2}), connect to one player in her subset and one player from the other subset (Figure~\ref{fig:bip3}), or to remove all her previous links and instead connect to all other players in her subset (Figure~\ref{fig:bip4}).
            	\vspace{-11pt}
            	\begin{figure}[hbt!]
                    \centering
                    \begin{subfigure}[t]{0.48\textwidth}
                        \centering
                        \begin{tikzpicture}[scale = 0.8]
                            %\tikzstyle{every node}=[font=\small]
                			\node[draw, circle, minimum size = 0.6cm,thick] at (0,0) (A0) {};
                			\node[draw, circle, minimum size = 0.6cm,thick] at (2,0) (A1) {};
                			\node[draw, circle, minimum size = 0.6cm,thick] at (4,0) (A2) {};
                			\node[draw, circle, minimum size = 0.6cm,thick] at (6,0) (A3) {};

                			\node[draw, circle, minimum size = 0.6cm,thick] at (1,2) (B0) {$u$};
                			\node[draw, circle, minimum size = 0.6cm,thick] at (3,2) (B1) {};
                			\node[draw, circle, minimum size = 0.6cm,thick] at (5,2) (B2) {};
                			
                			\draw[ ->,>=latex,thick] (B0) -- (A0);
                			\draw[ ->,>=latex,thick] (B0) -- (A1);
                			\draw[ ->,>=latex,thick] (B0) -- (A2);
                			\draw[ ->,>=latex,thick] (B0) -- (A3);
                			\draw[ ->,>=latex,thick] (B1) -- (A0);
                			\draw[ ->,>=latex,thick] (B1) -- (A1);
                			\draw[ ->,>=latex,thick] (B1) -- (A2);
                			\draw[ ->,>=latex,thick] (B1) -- (A3);
                			\draw[ ->,>=latex,thick] (B2) -- (A0);
                			\draw[ ->,>=latex,thick] (B2) -- (A1);
                			\draw[ ->,>=latex,thick] (B2) -- (A2);
                			\draw[ ->,>=latex,thick] (B2) -- (A3);
                		\end{tikzpicture}
                        \caption{strategy $s$\label{fig:bip1}}
                    \end{subfigure}
                    \hfill
                    \begin{subfigure}[t]{0.48\textwidth}
                        \centering
                    	\begin{tikzpicture}[scale = 0.8]
                            %\tikzstyle{every node}=[font=\small]
                			\node[draw, circle, minimum size = 0.6cm,thick] at (0,0) (A0) {};
                			\node[draw, circle, minimum size = 0.6cm,thick] at (2,0) (A1) {};
                			\node[draw, circle, minimum size = 0.6cm,thick] at (4,0) (A2) {};
                			\node[draw, circle, minimum size = 0.6cm,thick] at (6,0) (A3) {};

                			\node[draw, circle, minimum size = 0.6cm,thick] at (1,2) (B0) {$u$};
                			\node[draw, circle, minimum size = 0.6cm,thick] at (3,2) (B1) {};
                			\node[draw, circle, minimum size = 0.6cm,thick] at (5,2) (B2) {};
                			
                			\draw[ ->,>=latex,thick] (B0) -- (B1);
                			\draw[ ->,>=latex,thick] (B1) -- (A0);
                			\draw[ ->,>=latex,thick] (B1) -- (A1);
                			\draw[ ->,>=latex,thick] (B1) -- (A2);
                			\draw[ ->,>=latex,thick] (B1) -- (A3);
                			\draw[ ->,>=latex,thick] (B2) -- (A0);
                			\draw[ ->,>=latex,thick] (B2) -- (A1);
                			\draw[ ->,>=latex,thick] (B2) -- (A2);
                			\draw[ ->,>=latex,thick] (B2) -- (A3);
                		\end{tikzpicture}	
                    	\caption{strategy $\tilde{s}_1$ \label{fig:bip2}}
                    \end{subfigure}
                     \begin{subfigure}[t]{0.48\textwidth}
                        \centering
                    	\begin{tikzpicture}[scale = 0.8]
                            %\tikzstyle{every node}=[font=\small]
                			\node[draw, circle, minimum size = 0.6cm,thick] at (0,0) (A0) {};
                			\node[draw, circle, minimum size = 0.6cm,thick] at (2,0) (A1) {};
                			\node[draw, circle, minimum size = 0.6cm,thick] at (4,0) (A2) {};
                			\node[draw, circle, minimum size = 0.6cm,thick] at (6,0) (A3) {};
                			
                			\node[draw, circle, minimum size = 0.6cm,thick] at (1,2) (B0) {$u$};
                			\node[draw, circle, minimum size = 0.6cm,thick] at (3,2) (B1) {};
                			\node[draw, circle, minimum size = 0.6cm,thick] at (5,2) (B2) {};
                			
                			\draw[ ->,>=latex,thick] (B0) -- (A0);
                			\draw[ ->,>=latex,thick] (B0) -- (B1);
                			\draw[ ->,>=latex,thick] (B1) -- (A0);
                			\draw[ ->,>=latex,thick] (B1) -- (A1);
                			\draw[ ->,>=latex,thick] (B1) -- (A2);
                			\draw[ ->,>=latex,thick] (B1) -- (A3);
                			\draw[ ->,>=latex,thick] (B2) -- (A0);
                			\draw[ ->,>=latex,thick] (B2) -- (A1);
                			\draw[ ->,>=latex,thick] (B2) -- (A2);
                			\draw[ ->,>=latex,thick] (B2) -- (A3);
                		\end{tikzpicture}	
                    	\caption{strategy $\tilde{s}_2$ \label{fig:bip3}}
                    \end{subfigure}
                    \hfill
                     \begin{subfigure}[t]{0.48\textwidth}
                        \centering
                        \begin{tikzpicture}[scale = 0.8 ,->,>=stealth',auto]
                            %\tikzstyle{every node}=[font=\small]
                			\node[draw, circle, minimum size = 0.6cm,thick] at (0,0) (A0) {};
                			\node[draw, circle, minimum size = 0.6cm,thick] at (2,0) (A1) {};
                			\node[draw, circle, minimum size = 0.6cm,thick] at (4,0) (A2) {};
                			\node[draw, circle, minimum size = 0.6cm,thick] at (6,0) (A3) {};
                			
                			\node[draw, circle, minimum size = 0.6cm,thick] at (1,2) (B0) {$u$};
                			\node[draw, circle, minimum size = 0.6cm,thick] at (3,2) (B1) {};
                			\node[draw, circle, minimum size = 0.6cm,thick] at (5,2) (B2) {};
                			
                			\path[thick] (B0) edge [bend left] node {} (B2);
                			\draw[ ->,>=latex,thick] (B0) -- (B1);
                			\draw[ ->,>=latex,thick] (B1) -- (A0);
                			\draw[ ->,>=latex,thick] (B1) -- (A1);
                			\draw[ ->,>=latex,thick] (B1) -- (A2);
                			\draw[ ->,>=latex,thick] (B1) -- (A3);
                			\draw[ ->,>=latex,thick] (B2) -- (A0);
                			\draw[ ->,>=latex,thick] (B2) -- (A1);
                			\draw[ ->,>=latex,thick] (B2) -- (A2);
                			\draw[ ->,>=latex,thick] (B2) -- (A3);
                		\end{tikzpicture}
                    	\caption{strategy $\tilde{s}_3$ \label{fig:bip4}}
                    \end{subfigure}
                    \caption{\small Strategy deviations of player $u$.}
                \end{figure}
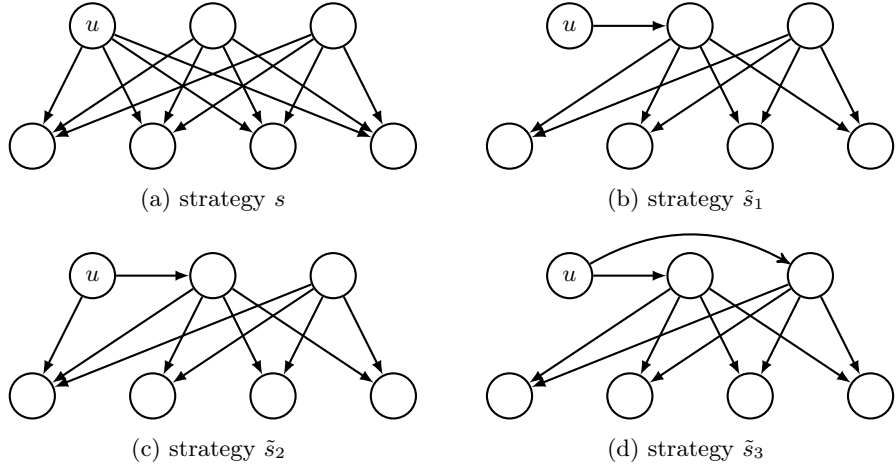
            % 	\vspace{-11pt}
            	When player $u$ changes to strategy $\tilde{s}_1$, seen in Figure~\ref{fig:bip2} the change in cost is as follows:  
                \begin{align*}
                    \Delta \text{cost}_u(s \text{ to } \tilde{s}_1) =& -(s-1) + \dfrac{s  \cdot (s-1)}{r} b + (s+r-3) \cdot c  
                \end{align*}
                as player $u$ initiates $s-1$ less links than before - losing all her previous betweenness. Additionally, she is one edge further away from all other players except for the one she connects to directly. 
                Thus, the above strategy is less preferable than the complete bipartite graph for player $u$, if
                $$1 \leq \dfrac{s}{r}  b + \dfrac{s+r -3}{s-1} c.$$
                Player $u$'s change to strategy $\tilde{s}_2$ (Figure~\ref{fig:bip3}) leads to $s-2$ less links initiated by her. The player is further away from $s-1$ players from the other subset and closer to one in her own. All transaction-routing potential is lost. Therefore, the change in cost is given by
                \begin{align*}
                    \Delta \text{cost}_u(s \text{ to } \tilde{s}_2) =& 2-s +  \left ( \dfrac{s \cdot (s-1)}{r} \right) b + (s-2)\cdot c. 
                \end{align*}
                Hence, for this strategy to be less preferable than the complete bipartite graph,
                $$1 \leq \left( \dfrac{s\cdot (s-1)}{r \cdot (s-2)} \right) b +  c =  \alpha \cdot b + c.$$
                When severing all previous links and connecting to all players in her subset instead, strategy $\tilde{s}_3$ (Figure~\ref{fig:bip4}), player $u$ builds $s-r+1$ less links than before. Furthermore, she is closer to players previously in her own subset and further away from the rest. While player $u$ can now transmit transactions of players previously in her own subset, she is no longer a preferable intermediary for players previously in the other subset.  Therefore, the change in cost is given by
                \begin{align*}
                    \Delta \text{cost}_u(s \text{ to } \tilde{s}_3) =& r- s+1 +  \left (\dfrac{s\cdot (s-1)}{r} -\dfrac{(r-1) (r-2)}{s+1}  \right) b + (s-r+1)\cdot c. 
                \end{align*}
                Hence, for this strategy to be less preferable than the complete bipartite graph for player $u$,
                $$1 \leq \dfrac{1}{(s-r+1)} \left( \dfrac{s\cdot (s-1)}{r} -\dfrac{(r-1) (r-2)}{s+1} \right) b +  c = \beta \cdot b +c.$$
                
                To summarize, the complete bipartite graph $K_{r,s}$ is a Nash equilibrium for 
                $$\frac{s-2}{r+1} b +c\leq 1 \leq \min \left\{ \frac{s}{r}  b + \dfrac{s+r-3}{s-1}  c,   \min \left\{\alpha,\beta \right\}\cdot b + c \right\}.$$ 
                \hfill \qed
            \end{proof}
            
            The parameter map for the complete bipartite graph is drawn in Figure~\ref{fig:biclique}. There $(\gamma, \delta )$ is the intersection between  $1 =  \frac{s}{r} b + \frac{s+r-3}{s-1} c$ and $1 = \min \left\{\alpha , \beta \right\} \cdot  b + c $.
            
        % \vspace{-11pt}    
        \subsubsection*{Simulation.}
            To better understand the behaviour of a player in our payment network creation game, we implement a simulation of the game~\cite{simulation}. Our simulation enumerates all Nash equilibria for a given number of players $n$, as well as the weights for the betweenness and closeness costs. However, this is only feasible for small $n$. Parameter sweeps for the weights $b$ and $c$ can also be performed to see when a given topology is a Nash equilibrium. Some parameter sweeps for topologies previously analyzed can be found in Appendix~\ref{app:simulation}. Finally, starting from an initial graph the progression of the game can be simulated.
        
    \subsection{Price of Anarchy}
        The ratio between the social optimum and the worst Nash equilibrium is the price of anarchy ($\text{PoA}$), formally,
        $$\text{PoA} = \dfrac{\max _{s\in N} \text{cost}(s)}{\min _{s \in S} \text{cost}(s)},$$
        here $S$ is the set of all strategies and $N$ is the set of strategies that are Nash equilibria.     
        
        The price of anarchy provides an insight to the effects of lack of coordination, i.e.\ measures the performance degradation of the system when players act selfishly in comparison to central coordination.
        When the price of anarchy is low, selfish actors do not heavily degrade network efficiency. In contrast, a high price of anarchy indicates that network formation by a central authority would significantly increase efficiency. 
        
        For $c>1$, we can determine the price of anarchy exactly, as we established both the social optimum and the (unique) Nash equilibria for $c>1$.
        %have found the social optimum for our entire parameter space and are aware of all Nash equilibria for $c>1$.
        \begin{restatable}[]{corollary}{poaone}\label{cor:poa1}
            For $ c > 1$ and $c > \frac{1}{2} +b$, the price of anarchy is $\text{PoA} =1$.
        \end{restatable}
        %  \begin{proof}
        %         The only Nash equilibrium for $c>1$ is the complete graph as stated by 
        %         Theorem~\ref{thm:thridtheorem}. As the social optimum for $c > \frac{1}{2} +b$ is also the complete graph (Theorem~\ref{thm:firsttheorem}), the price of anarchy is 
        %         $\text{PoA} =1,$
        %         for $ c > 1$ and $c > \frac{1}{2} +b$. \hfill \qed
        %     \end{proof}
        
        \begin{restatable}[]{corollary}{poatwo}\label{cor:poa2}
            For $c > 1$ and $b \leq c \leq \frac{1}{2} +b$, the price of anarchy is $$\text{PoA} = \frac{\left(\frac{1}{2}+ (n-2) \cdot b \right)\cdot n}{1+ (c + b\cdot (n-1)) (n-2)}.$$
        \end{restatable}
            %  \begin{proof}
            %     For $c > 1$ and $b \leq c \leq \frac{1}{2} +b$ the only Nash equilibrium is the complete graph (Theorem~\ref{thm:thridtheorem}) and according to Theorem~\ref{thm:firsttheorem}, the social optimum is the star graph. Thus, the price of anarchy is given by 
            % \begin{align*}
            %     \text{PoA} =& \dfrac{\text{cost}(\text{complete graph})}{\text{cost}(\text{star graph})}
            %     = \dfrac{\left(\frac{1}{2}+ (n-2) \cdot b \right)(n-1)\cdot n}{(1-2(c-b)+(c-b)\cdot n+b\cdot (n-2)\cdot n)(n-1)}\\
            %     =& \dfrac{(\left(\frac{1}{2}+ (n-2) \cdot b \right)\cdot n}{1+ (c + b\cdot (n-1))(n-2)}
            % \end{align*} \hfill \qed
            % \end{proof}

        \begin{restatable}[]{corollary}{poathree}\label{cor:poa3}
            For $1< c <b$ , the price of anarchy is $$\text{PoA}= \frac{\left(\frac{1}{2}+ (n-2) \cdot b \right)\cdot n}{1  + \left(\frac{2}{3}b + \frac{1}{3}c\right) \cdot n \cdot (n-2)} .$$
        \end{restatable}
        % \begin{proof}
        %     For $1<c<b$ the only Nash equilibrium is the complete graph as stated in Theorem~\ref{thm:thridtheorem} and the social optimum is a path graph (Theorem~\ref{thm:firsttheorem}). The price of anarchy is given by 
        %     \begin{align*}
        %         \text{PoA} =& \dfrac{\text{cost}(\text{complete graph})}{\text{cost}(\text{path graph})}
        %         = \dfrac{\left(\frac{1}{2}+ (n-2) \cdot b \right)(n-1)\cdot n}{\left( 1  + \left(\frac{2}{3}b + \frac{1}{3}c\right) \cdot n \cdot (n-2)\right) (n-1)}=\\
        %         =& \dfrac{(\left(\frac{1}{2}+ (n-2) \cdot b \right)\cdot n}{1  + \left(\frac{2}{3}b + \frac{1}{3}c\right) \cdot n \cdot (n-2) }
        %     \end{align*}\hfill \qed
        % \end{proof}

        Combining the results of Corollary~\ref{cor:poa1},~\ref{cor:poa2} and~\ref{cor:poa3} allows us to upper bound the price of anarchy to a constant for $c>1$, as stated in in Corollary~\ref{cor:poa4}. This upper bound is asymptotically tight, as the price of anarchy is always greater or equal to one (hence at least constant) by definition.   
            
        \begin{restatable}{corollary}{poafour}\label{cor:poa4}
           For $c > 1$, the price of anarchy is $\text{PoA} = \mathcal{O}(1)$.
        \end{restatable}
        % \begin{proof}
        %     For $c>1$ and $c > \frac{1}{2}+b$, the price of anarchy is one and therefore it is also $\mathcal{O}(1)$. 
                
        %     We have that for $c>1$ and $b \leq c \leq \frac{1}{2} +b$,  
        %     $$\text{PoA} = \frac{\left(\frac{1}{2}+ (n-2) \cdot b \right)\cdot n}{1+ (c + b\cdot (n-1)) (n-2)} = \mathcal{O} \left(\frac{b\cdot n^2}{ b\cdot n^2} \right)  =  \mathcal{O}(1),$$
        %     and for $1<c<b$, 
        %     $$\text{PoA} = \frac{\left(\frac{1}{2}+ (n-2) \cdot b \right)\cdot n}{1  + \left(\frac{2}{3}b + \frac{1}{3}c\right) \cdot n \cdot (n-2)} =\mathcal{O} \left(\frac{b\cdot n^2}{ b\cdot n^2} \right)  =  \mathcal{O}(1).$$
        %     Thus, for $c>1$ we have $\text{PoA} =\mathcal{O}(1)$. \hfill \qed
        % \end{proof}  
        
        For small $b$ and $c$ we can also upper bound the price of anarchy as follows:
        \begin{restatable}{theorem}{poafive}\label{thm:poa}
            For $c + b < \frac{1}{n^2}$, the price of anarchy is $\text{PoA} = \mathcal{O}(1)$. 
        \end{restatable}
        % \begin{proof}
        %     For $c+b <\frac{1}{n^2}$, all Nash equilibria are trees. Unless the distance to a player is infinite, no player in the network will have an incentive to build an edge. 
                    
        %     As both the maximum possible change in  $\text{betweenness}_u(s)$ and $\text{closeness}_u(s)$ for a node $u$ in a connected graph is less than $n^2$ and all Nash equilibria are connected, $\Delta \text{cost}_u (s) > - n ^2 \cdot c -  n ^2 \cdot b  +1$.
        %     We require $\Delta \text{cost}_u (s) \geq 0$ such that $u$ does not benefit from initiating an additional channel. Thus, for $c + b \leq \frac{1}{n^2}$ all Nash equilibria are spanning trees. 
                    
        %     For $c + b \leq \frac{1}{n^2}$  the social optimum is also a spanning tree, as it is either the star or path graph. It easily follows that for $c + b \leq \frac{1}{n^2}$ and all spanning trees $\text{cost}(s) = \Theta(n)$ and therefore the price of anarchy is $\mathcal{O}(1)$.
        % \end{proof}

        Finally, for $c+b \geq \frac{1}{n^2}$ and $c<1$, we show an $\mathcal{O}(n)$ upper bound for the price of anarchy. 
            
        \begin{restatable}{theorem}{poasix}
            For $c + b \geq \frac{1}{n^2}$ and $c<1$, the price of anarchy is $\text{PoA} = \mathcal{O}(n)$. 
        \end{restatable}

    \subsection{Price of Stability}  
        The price of stability ($\text{PoS}$), a close notion to price of anarchy, is defined as the ratio between the social optimum and the best Nash equilibrium, 
        $$\text{PoS} = \dfrac{\min _{s\in N} \text{cost}(s)}{\min _{s\in S} \text{cost}(s)},$$
        where $S$ is the set of all strategies and $N$ is the set of strategies that are Nash equilibria. The price of stability expresses the loss in network performance in stable systems in comparison to those designed by a central performance. Corollary \ref{cor:pos} gives insight into the price of stability in regions of the parameter space previously discussed in the context of the price of anarchy.
        
        \begin{restatable}{corollary}{pos}\label{cor:pos}
           For $c > 1$ and $b + c< \frac{1}{n^2}$, the price of stability $\text{PoS} = \mathcal{O}(1)$. 
        \end{restatable}
        % \begin{proof}
        %     As the price of stability is smaller than or equal to the price of anarchy, we can follow from Corollary \ref{cor:poa4}, that the price of stability is $\mathcal{O}(1)$ for $c>1$. Additionally, Theorem \ref{thm:poa} indicates that $\text{PoS} = \mathcal{O}(1)$ for $b + c< \frac{1}{n^2}$.
        %     \hfill \qed
        % \end{proof}
    
        However, we expect blockchain payment networks to fall into the remaining area, where $c + b \geq \frac{1}{n^2}$ and $c<1$. 
        In particular, considering the underlying uniform transaction scenario and the fixed blockchain fee equal to one (wlog), a competitive transaction fee would be $\frac{1}{n}$. Thus, an appropriate allocation for the weights is $b = \frac{1}{2n}$ and $c = \frac{1}{n}$, as the betweenness term counts each sender and receiver pair twice. 
        For these weights the star is the social optimum (Theorem \ref{thm:firsttheorem}), as well as a Nash equilibrium (Theorem \ref{thm:star}). Hence, the price of stability for payment networks is one; indicating that an optimal payment network is stable in a game with selfish players. Thus, payment networks can be stable and efficient.

 \section{Conclusion}

 We introduced a game-theoretic model to encapsulate the creation of  payment networks. To this end, we generalized previous work, as our model is more complex and demands a combination of betweenness and closeness centralities that have thus far only been studied independently in network creation games. 
        
First, we identified the social optimum for the entire parameter space of our game. Depending on the weights placed on the betweenness and closeness centralities either the complete graph, the star graph or the path graph is the social optimum. In the area of the parameter space that most accurately reflects payment networks, we found the star graph to be the social optimum.
        
Next, we examined the space of possible Nash equilibria. After establishing that finding the best response of a player is NP-hard, we analyzed prominent graphs and determined if and when they constitute a Nash equilibrium. We showed that the complete graph is the only Nash equilibrium if players place a large weight on their closeness centrality; reflecting payment channels in which players execute many transactions or value privacy highly.
On the other hand, both the path and circle graph are Nash equilibria only for small number of players and thus are not expected to emerge as stable structures in payment networks. On the contrary, the star graph emerges as a Nash equilibrium for the areas of our parameter space most accurately representing payment networks. In addition, we observed that depending on the size of the subsets, the complete bipartite graph is also a Nash equilibrium in similar regions of the parameter space as the star graph.
        
Last, combining our results, we bounded the price of anarchy for a large part of the parameter space. In particular, we proved that when the closeness centrality weight is high, meaning that the players execute transactions frequently or demand privacy, the price of anarchy is constant; indicating little loss in network performance for selfish players. 
On the other hand, for small weight on the closeness centrality, we showed an $\mathcal{O}(n)$ upper bound on the price of anarchy. Nevertheless, the price of stability in payment networks is equal to one, since the star is both the social optimum and a Nash equilibrium for suitable parameters; demonstrating that blockchain payment networks can indeed be both stable and efficient, when forming more centralized network structures.

\newpage
 \bibliographystyle{splncs04}
 \bibliography{references}

\newpage
% This creates an appendix chapter, comment if not needed.
\appendix
    \section{Parameter Sweeps}\label{app:simulation}
        We show some parameter sweeps for the weights $b$ and $c$ generated by our simulation. These show when some of the prominent graphs analyzed in Chapter~\ref{chap:Nash} are Nash equilibria. 
        \subsection{Complete Graph}
            Figure~\ref{fig:completesimulation} shows the simulation results for the complete graph. Here the underlying assumption was made that lower ID players connected to all higher ID players. However, independent of this assumption the simulation yields the same results. 
            \begin{figure}[h!]
                \centering
                \begin{subfigure}[t]{0.48\textwidth}
                    \centering
                    \includegraphics[scale=0.57]{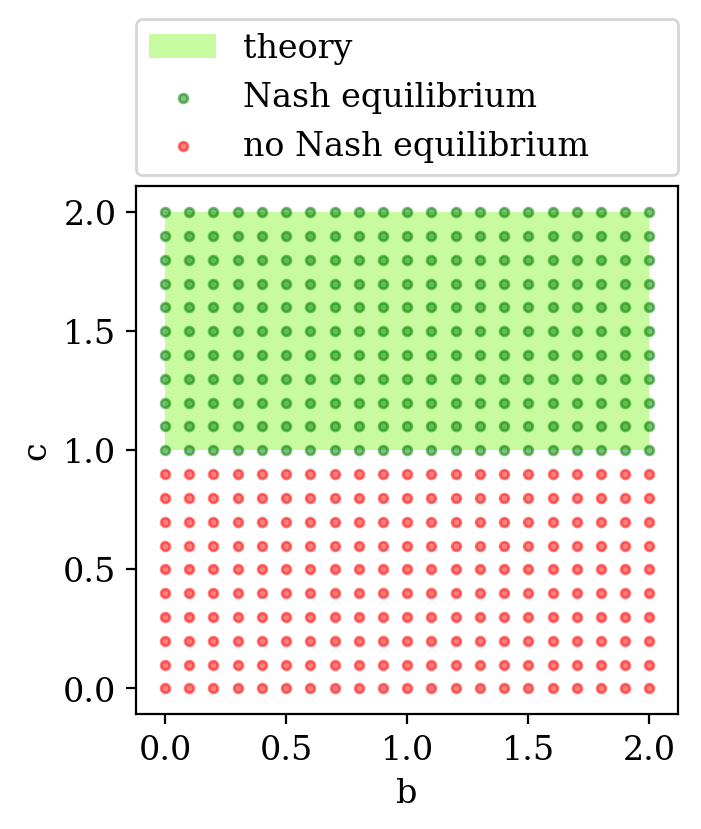}
                    \caption{$n=3$ \label{fig:completen3}}
                \end{subfigure}
                \hfill
                \begin{subfigure}[t]{0.48\textwidth}                  
                    \centering
                    \includegraphics[scale=0.57]{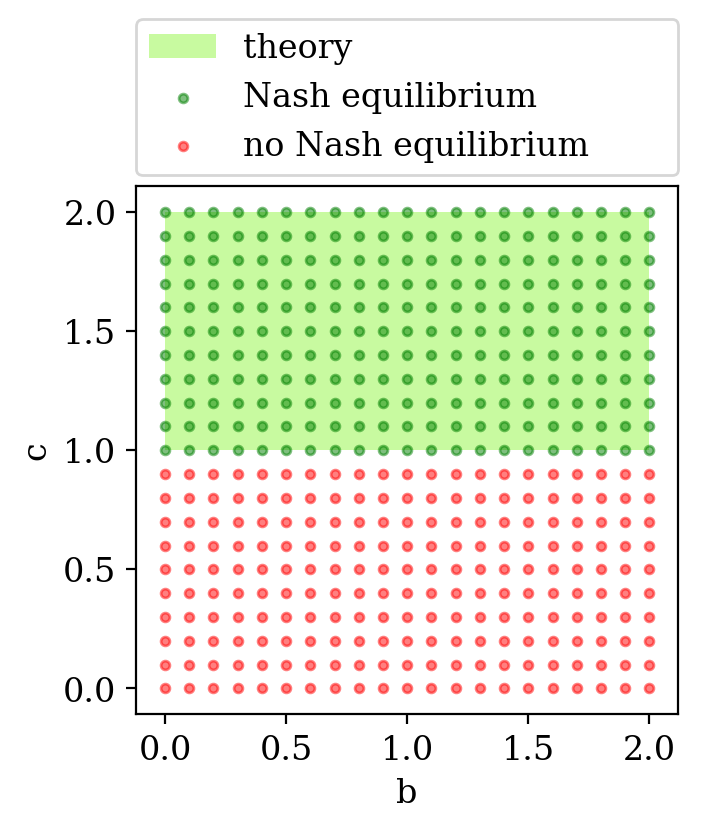}
                    \caption{$n=8$ \label{fig:completen8}}
                \end{subfigure}
                \caption{\small Parameter map for complete graph.}
                \label{fig:completesimulation}
            \end{figure}
        \subsection{Path Graph}
            Simulations done for the path graph are shown in Figure~\ref{fig:pathsimulation}. The assignment of outgoing edges is such that the least restrictive bounds are achieved. 
            \begin{figure}[h!]
                \centering
                \begin{subfigure}[t]{0.48\textwidth}
                    \centering
                    \includegraphics[scale=0.57]{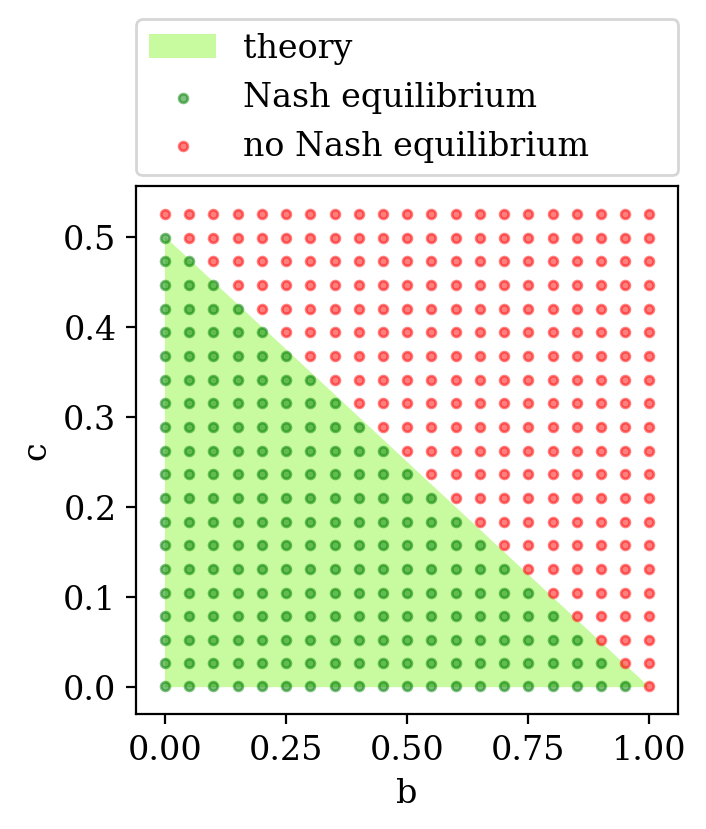}
                    \caption{$n=4$ \label{fig:pathn4}}
                \end{subfigure}
                \hfill
                \begin{subfigure}[t]{0.48\textwidth}                  
                    \centering
                    \includegraphics[scale=0.57]{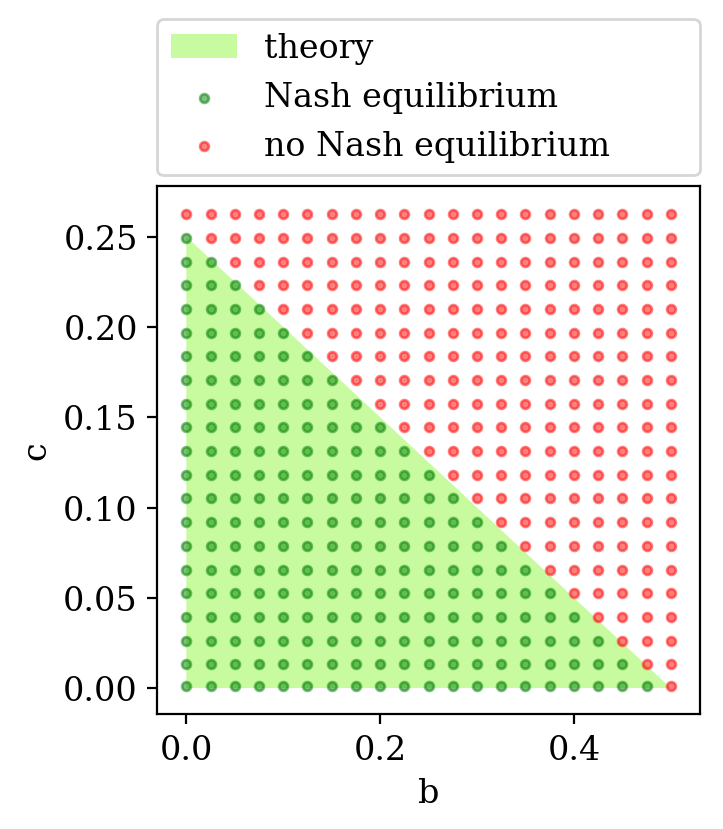}
                    \caption{$n=5$ \label{fig:pathn5}}
                \end{subfigure}\vspace{0.2cm}
                \begin{subfigure}[t]{0.48\textwidth}
                    \centering
                    \includegraphics[scale=0.57]{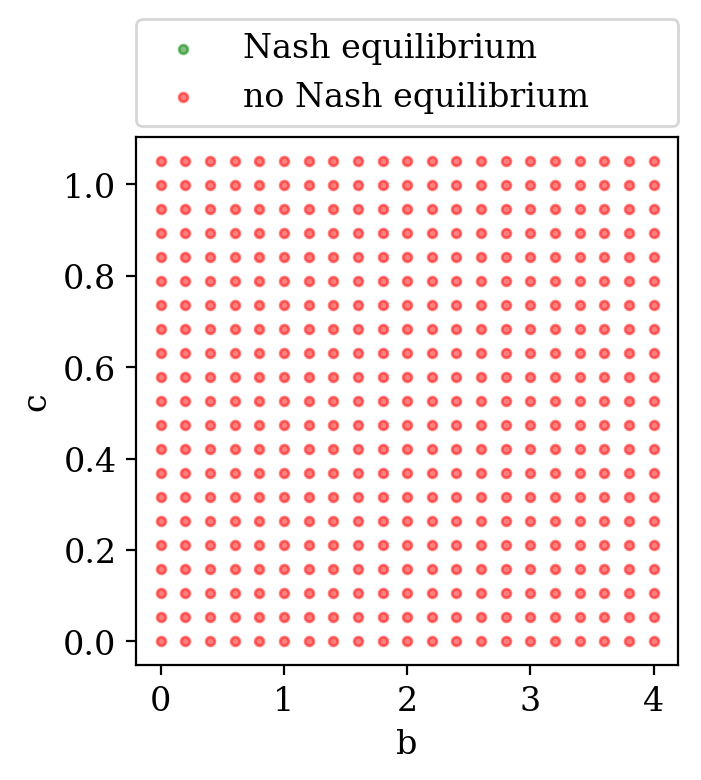}
                    \caption{$n=6$ \label{fig:pathn6}}
                \end{subfigure}
                \hfill
                \begin{subfigure}[t]{0.48\textwidth}                  
                    \centering
                    \includegraphics[scale=0.57]{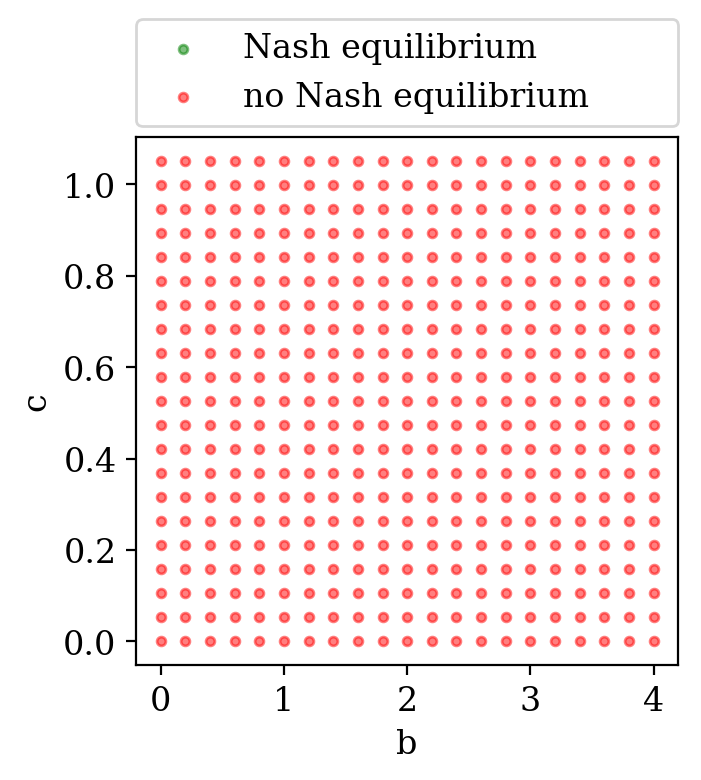}
                    \caption{$n=7$ \label{fig:pathn7}}
                \end{subfigure}
                \caption{\small Parameter map for path graph.}
                \label{fig:pathsimulation}
            \end{figure}
        \subsection{Circle Graph}\label{app:simulationcircle}
            Simulation results for the circle graph are shown in Figure~\ref{fig:circlesimulation}. Here all players have exactly one outgoing link, as the bounds for this case are less restrictive than if any player would have two outgoing links. 
            \begin{figure}[h!]
                \centering
                \begin{subfigure}[t]{0.48\textwidth}
                    \centering
                    \includegraphics[scale=0.57]{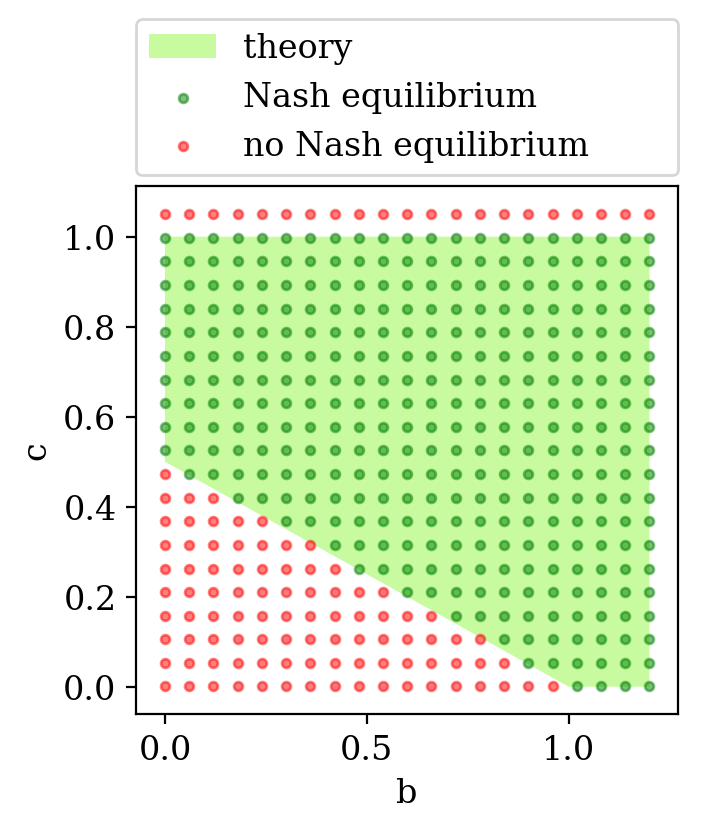}
                    \caption{$n=4$ \label{fig:circlen4}}
                \end{subfigure}
                \hfill
                \begin{subfigure}[t]{0.48\textwidth}                  
                    \centering
                    \includegraphics[scale=0.57]{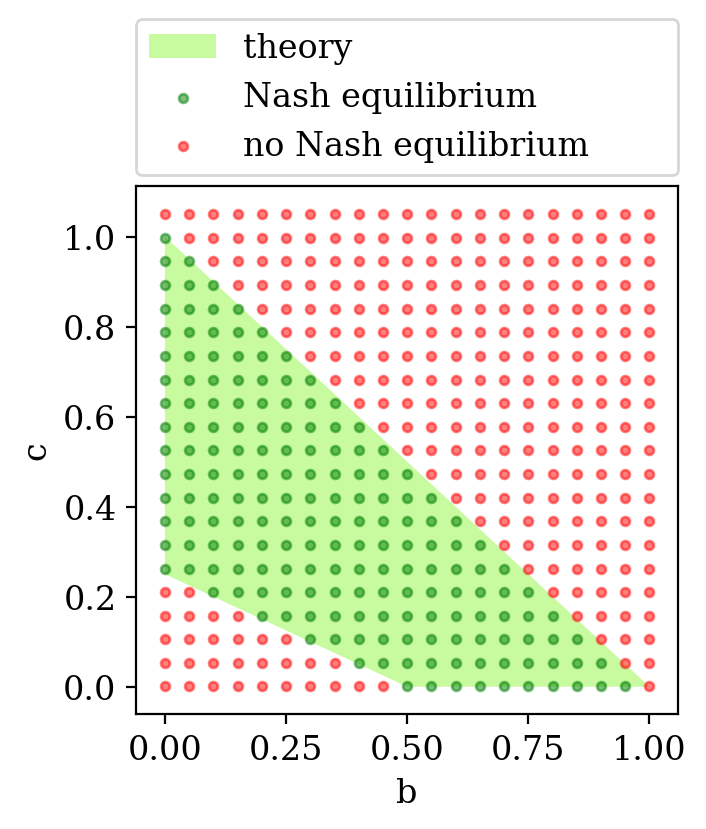}
                    \caption{$n=5$ \label{fig:circlen5}}
                \end{subfigure}\vspace{0.2cm}
                \begin{subfigure}[t]{0.48\textwidth}
                    \centering
                    \includegraphics[scale=0.57]{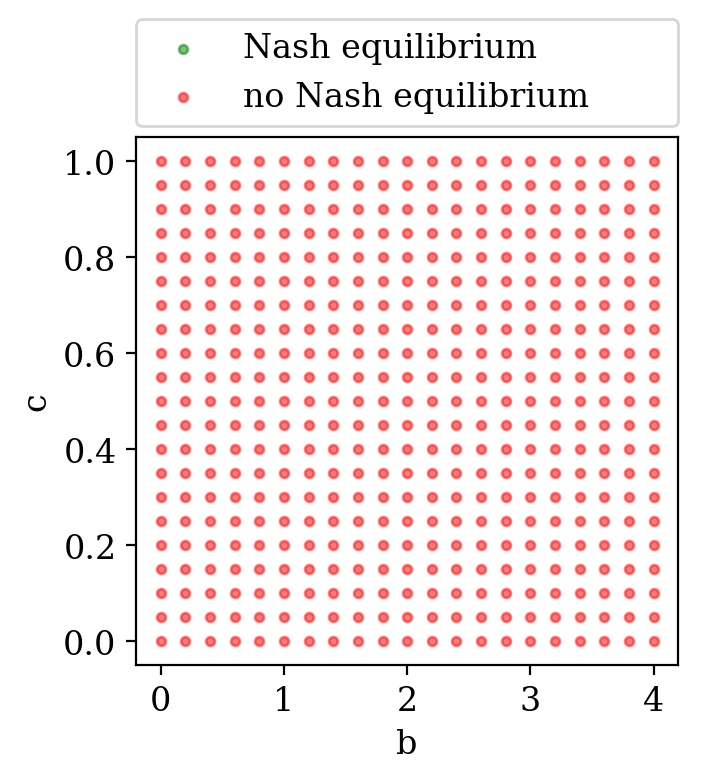}
                    \caption{$n=6$ \label{fig:circlen6}}
                \end{subfigure}
                \hfill
                \begin{subfigure}[t]{0.48\textwidth}                  
                    \centering
                    \includegraphics[scale=0.57]{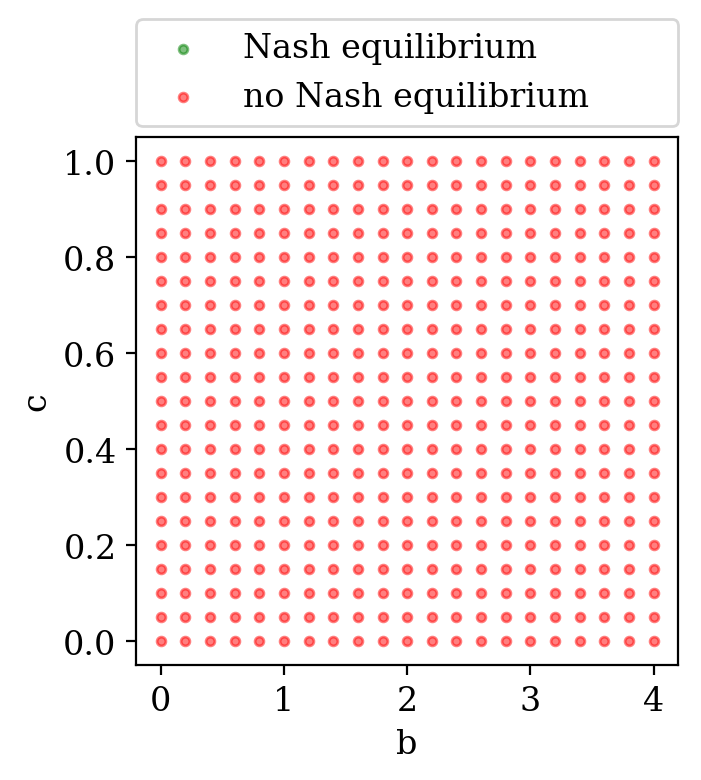}
                    \caption{$n=7$ \label{fig:circlen7}}
                \end{subfigure}
                \caption{\small Parameter map for circle graph.}
                \label{fig:circlesimulation}
            \end{figure}
            
            For $n=4$ and $n=5$, the simulation matches the theory. Additionally, the simulation also suggested that for $n\geq 6$ the circle graph is never a Nash equilibrium as indicated by the results in Figures~\ref{fig:circlen6} and~\ref{fig:circlen7}.  
        
        \subsection{Star Graph}
            In Figure~\ref{fig:starsimulation} we show when the star is a Nash equilibrium. For the simulation one player connects to everyone else. Other possible star graphs have more restrictive bounds on $b$ and $c$. 
            \begin{figure}[hbt!]
                \centering
                \begin{subfigure}[t]{0.48\textwidth}
                    \centering
                    \includegraphics[scale=0.57]{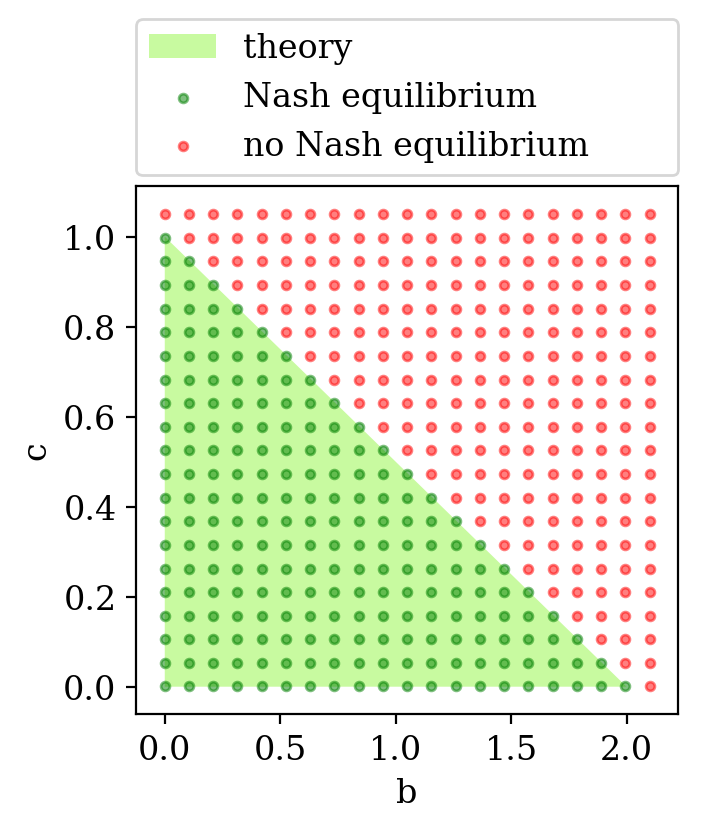}
                    \caption{$n=4$ \label{fig:star4}}
                \end{subfigure}
                \hfill
                \begin{subfigure}[t]{0.48\textwidth}                  
                    \centering
                    \includegraphics[scale=0.57]{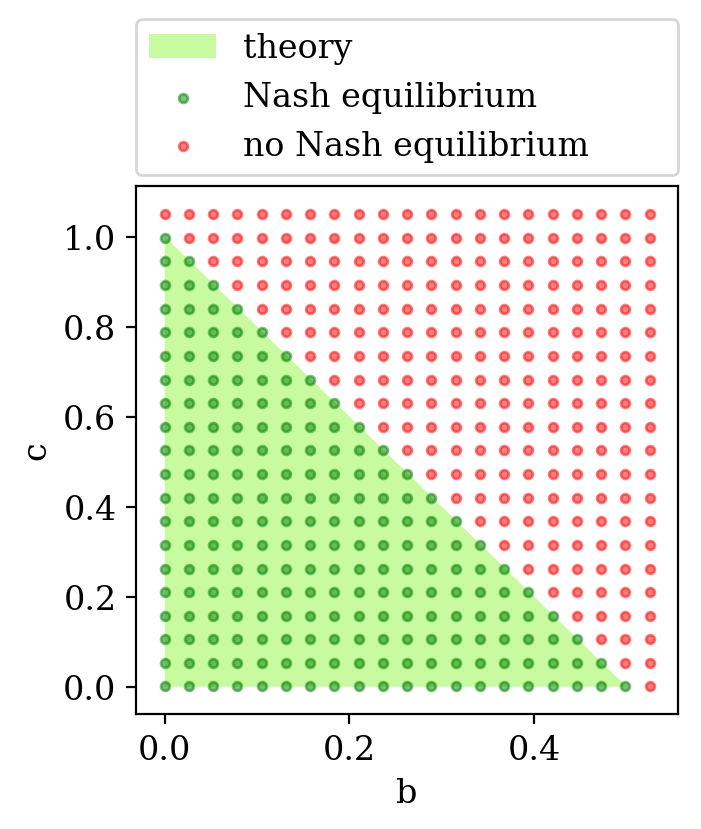}
                    \caption{$n=7$ \label{fig:star7}}
                \end{subfigure}
                \caption{\small Parameter map for star graph.}
                \label{fig:starsimulation}
            \end{figure}
        \subsection{Complete Bipartite Graph}
            Figure~\ref{fig:bicliquesimulation} shows when the complete bipartite graph is Nash equilibrium for the presented parameters. The simulation was done for players in the smaller subset having all the outgoing links, as this case leads to less restrictive bounds for $b$ and $c$. 
            \begin{figure}[hbt!]
                \centering
                \begin{subfigure}[t]{0.48\textwidth}
                    \centering
                    \includegraphics[scale=0.57]{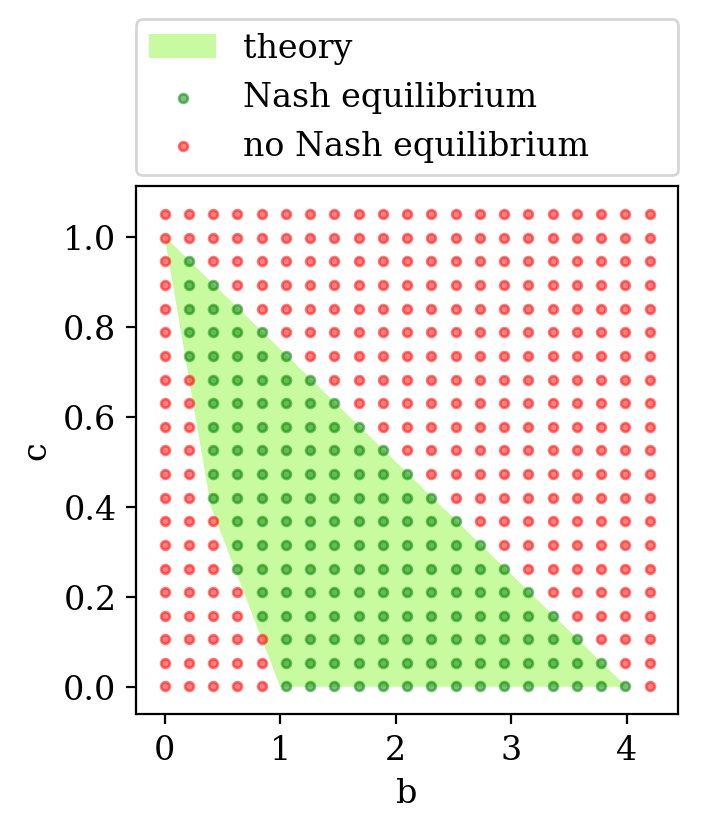}
                    \caption{$n=6$, $r=3$, $s=3$ \label{fig:n6r3s3}}
                \end{subfigure}
                \hfill
                \begin{subfigure}[t]{0.48\textwidth}
                    \centering
                    \includegraphics[scale=0.57]{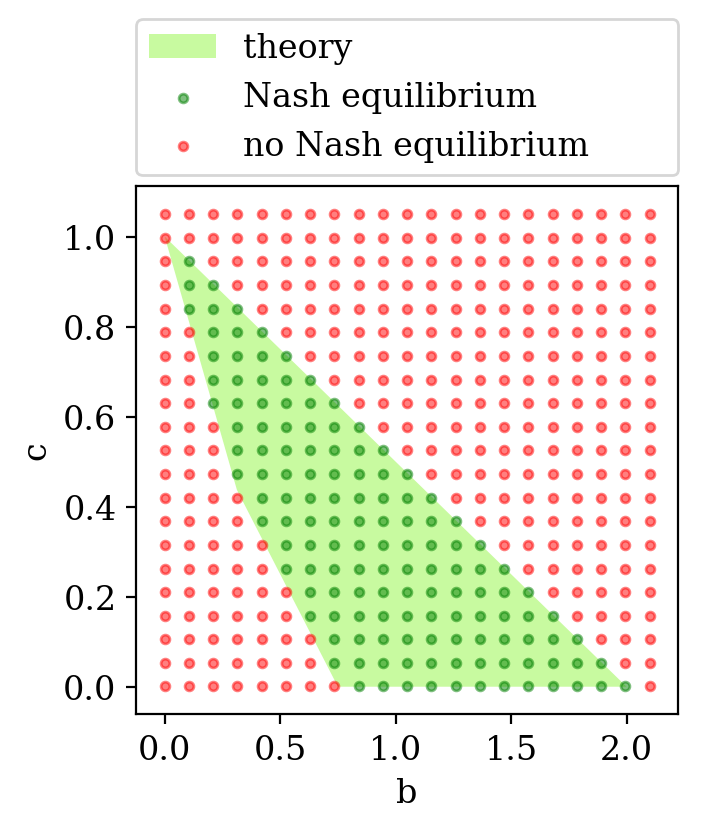}
                    \caption{$n=7$, $r=3$, $s=4$ \label{fig:n7r3s4}}
                \end{subfigure}\vspace{0.2 cm}
                \begin{subfigure}[t]{0.48\textwidth}
                    \centering
                    \includegraphics[scale=0.57]{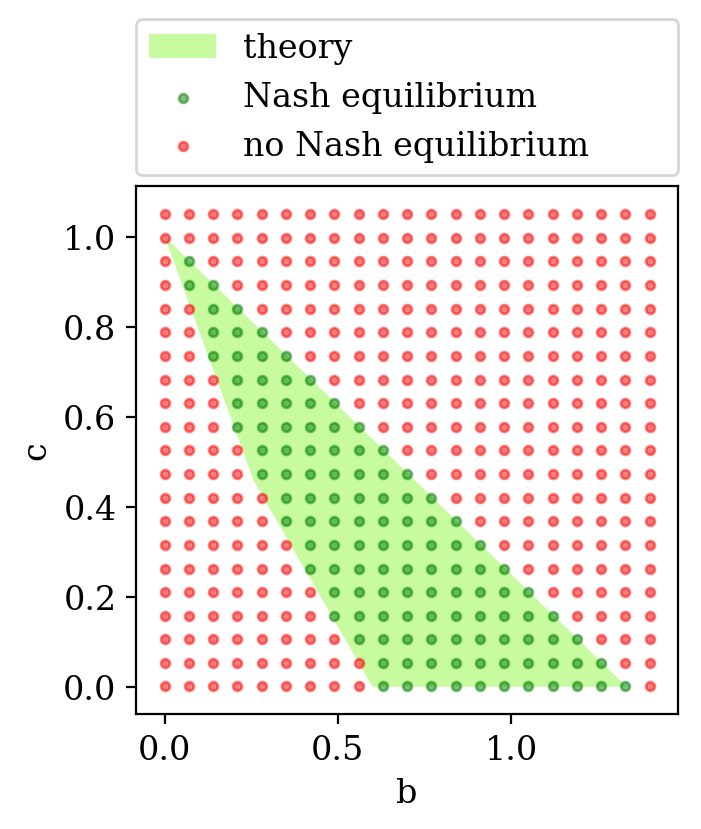}
                    \caption{$n=8$, $r=3$, $s=5$ \label{fig:n8r3s5}}
                \end{subfigure}
                \hfill
                \begin{subfigure}[t]{0.48\textwidth}
                    \centering
                    \includegraphics[scale=0.57]{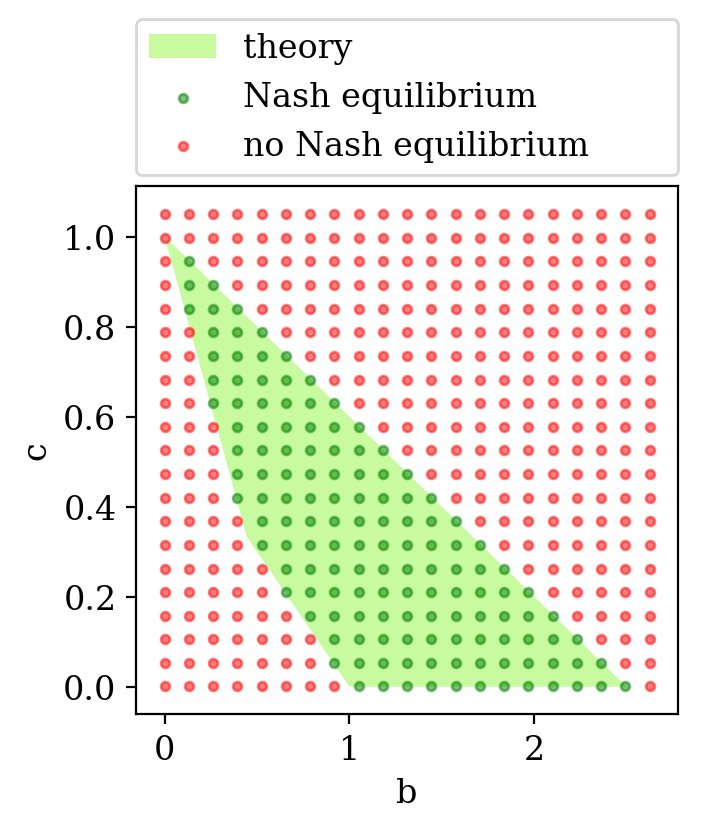}
                    \caption{$n=8$, $r=4$, $s=4$ \label{fig:n8r4s4}}
                \end{subfigure}
                \caption{\small Parameter map for complete bipartite graph.}
                \label{fig:bicliquesimulation}
            \end{figure}
            
    \section{Proofs} \label{app:proofs}
        \subsection{Social Optimum}
            \begin{lemma}[Theorem 1~\cite{gago2007betweenness}]\label{lem:firstlemma}
                The average betweenness $\overline{B}(G)$ in a connected graph $G$ can be expressed as: $\overline{B}(G) = (n-1) ( \overline{l}(G)-1)$, where $\overline{l}(G)$ is the average distance in $G$.
            \end{lemma}
            Lemma~\ref{lem:firstlemma} is proven in~\cite{gago2007betweenness} and relates the average betweenness and distance in a connected graph. We take advantage of Lemma~\ref{lem:firstlemma} to simplify the social cost expression.   
            \socialcost*
            \begin{proof}
                According to Lemma~\ref{lem:firstlemma} the social cost can be expressed as  follows for all $b \geq 0$ and $c >0$. 
                \begingroup
                    \allowdisplaybreaks
                    \begin{align*}
                        \text{cost}(s) =& \lvert E(G) \rvert  + b  \sum \limits _{u \in [n]} \text{betweenness} (u) + c \sum \limits _{u \in [n]} \text{closeness} (u) \\
                        = &\lvert E(G) \rvert  + b  \sum \limits _{u \in [n]} \left((n-1)(n-2) -\sum \limits _{\substack{s,r \in [n]: \\  s \neq r\neq u, \\ m(s,r) > 0}} \dfrac{m_u (s,r)}{m (s,r)}\right) \\
                        &+ c  \sum \limits _{u \in [n]}  \sum \limits _{r \in [n]-u} \left( d_{G[s]}(u,r) -1 \right)\\
                        =&\lvert E(G)\rvert + b \cdot n\cdot (n-1)(n-2) - b \cdot n \cdot \overline{B}(G) + c \cdot n \cdot (n-1)  (\overline{l}(G) -1) \\
                        =& \lvert E(G) \rvert + b \cdot n\cdot (n-1)(n-2)+ (c- b) \cdot n \cdot (n-1) (\overline{l}(G) -1)\\
                        =& \lvert E(G) \rvert  + b \cdot n\cdot (n-1)(n-2)+ (c- b) \cdot \sum \limits _{u \in [n]}  \sum \limits _{r \in [n]-u} \left( d_{G[s]}(u,r) -1 \right) \hfill \qed
                    \end{align*}
                \endgroup
            \end{proof}
            
    % \subsection{Social Optimum}
            \socialopt*
             \begin{proof}
            Using Lemma~\ref{lem:secondlemma} we can lower bound the social cost for $ c \geq b$ as follows:
            \begin{align*}
                \text{cost}(s) =& \lvert E(G) \rvert+ b \cdot n\cdot (n-1)(n-2)+ \underbrace{(c- b)}_{\geq 0}  \sum \limits _{u \in [n]}  \sum \limits _{r \in [n]-u} \left( d_{G[s]}(u,r) -1 \right) \\
                \geq & \lvert E(G)\rvert + b \cdot n\cdot (n-1)(n-2)+ (c-b) (n\cdot (n-1) -2\lvert E(G) \rvert )\\
                =& (1- 2\cdot (c-b)) \cdot \lvert E(G)\rvert + b \cdot n\cdot (n-1)(n-2)+ (c-b) (n\cdot (n-1))
            \end{align*}
            
            since every pair of nodes that is not connected by an edge is at least distance 2 apart~\cite{fabrikant2003on}. This lower bound is achieved by any graph with diameter at most 2. It follows that for $c > \frac{1}{2} + b$ the social optimum is a complete graph, maximizing $\lvert E \rvert$, and for $b \leq c \leq \frac{1}{2} + b$ the social optimum is a star, minimizing $\lvert E \rvert$.
            
            To find the social optimum for $c <b$, we rewrite the social cost as
            \begin{align*}
                \text{cost}(s) =& \lvert E(G) \rvert+ b \cdot n\cdot (n-1)(n-2) - (b- c) \cdot \sum \limits _{u \in [n]}  \sum \limits _{r \in [n]-u} \left( d_{G[s]}(u,r) -1 \right) \\
                =& \lvert E(G) \rvert- 2\cdot (b- c) \cdot d(G)+ b \cdot n\cdot (n-1)(n-2)  + (b-c)  \cdot n \cdot (n-1)
            \end{align*}
            For a connected graph the social cost is then minimized for a tree, as 
            $ \lvert E(H) \rvert - a \cdot d(H) > \lvert E(G) \rvert - a \cdot d(G) $
            if $G$ is a subgraph of $H$ and $a > 0$. For any tree, the number of edges is $n-1$. 
            Using Lemma~\ref{lem:thirdlemma}, we get that 
            \begin{align*}
                \text{cost}(s) =& \lvert E(G) \rvert+ b \cdot n\cdot (n-1)(n-2) - (b- c)  \sum \limits _{u \in [n]}  \sum \limits _{r \in [n]-u} \left( d_{G[s]}(u,r) -1 \right) \\
                \geq & \left( 1 + \left( \dfrac{2}{3}b +\dfrac{1}{3}c \right)  n \cdot (n-2) \right) (n-1)
            \end{align*}
            is a lower bound for the social cost which is achieved by a path graph.\hfill \qed 
        \end{proof} 
        
        \subsection{Hardness of Finding the Best Response}
            \np*
            
            \begin{proof} 
                The following proof is adapted from Proposition 1 in~\cite{fabrikant2003on}.
                
                Given the configuration of the rest of the graph, player $u$ has to compute her best response; a subset of players to build channels to such that her cost is minimized. For $b =0$ and $0.5 < c < 1$ and no incoming links from the rest of the graph, we know that the diameter of $G$ can be at most 2. Additionally, making more than the minimum number of required links, only improves the distance term by $c$, which is strictly smaller than the cost of establishing a link. Thus, $u$'s strategy is a dominating set for the rest of the graph. 
                    
                The cost of $u$ is minimized when the size of the subset is minimized. The minimum size dominating set corresponds to $u$'s best response. Hence, it is NP-hard to compute a player's best response by reduction from the dominating set. \hfill \qed 
            \end{proof}
            
        \subsection{Complete graph}
            \cliqueone*
            \begin{proof}
                The addition of an edge by a player never increases her betweenness cost. Thus, by the definition of the cost function any Nash equilibrium cannot be missing any edges whose addition would reduce a players closeness by more than 1, the cost of building an edge. As $c>1$, no edge can be missing in the graph and the only Nash equilibrium is the complete graph.\hfill \qed 
            \end{proof}
            
           \cliquetwo*
            \begin{proof}
                In a complete graph the removal of an edge by a player does not change her betweenness cost and her closeness cost is increased by $c$. Thus, the cost of a player would decrease when removing one edge. Therefore, the complete graph is not a Nash equilibrium for $c<1$.\hfill \qed 
            \end{proof}

        % \subsection{Path Graph for Small $n$}
        \subsection{Path Graph}
            \patha*
            \begin{proof}
                In a game with four players, a path graph with outgoing links from the endpoints is never a Nash equilibrium. Such an endpoint could simply reduce her cost by $c$ through exchanging her current channel with a channel to the player currently two edges away.
                
                In the remaining path graphs, all channels are initiated by the central nodes. A player can never increase her cost by removing or exchanging edges. The minimum change of cost can be achieved by an endpoint initiating a channel to the other endpoint. This change in cost is given by 
                $$\Delta \text{cost}(\text{add $1$ link}) = 1 - b - 2 \cdot c.$$
                We follow that for $n=4$, the path graph is a Nash equilibrium for $1 \leq b + 2 \cdot c$.
                \hfill \qed 
            \end{proof}
            
            \pathb*
            \begin{proof}
                As in the case with four players, a path graph with five nodes in which the endpoints initiate channels is never a Nash equilibrium. For instance, an endpoint could reduce her cost by $2\cdot c$ through connecting to the player currently three edges away instead. 
                
                Additionally, in a path graph with five players, the players neighboring the endpoints cannot initiate channels to the central node in a Nash equilibrium. Replacing such a link with a link to the other neighbor of the central node would lead to a cost reduction of $c$. 
                
                Thus, it only remains to consider a path graph with channels initiated as show in Figure~\ref{fig:pathwith5}.
                
                \begin{figure}
                    \centering
                    \begin{tikzpicture}[scale = 0.9]
                        \tikzstyle{every node}=[font=\scriptsize]
                		\node[draw, circle, minimum size = 0.6cm,thick] at (0,0) (A0) {};
                		\node[draw, circle, minimum size = 0.6cm,thick] at (2,0) (A1) {};
                		\node[draw, circle, minimum size = 0.6cm,thick] at (4,0) (A2) {};
                		\node[draw, circle, minimum size = 0.6cm,thick] at (6,0) (A3) {};
                		\node[draw, circle, minimum size = 0.6cm,thick] at (8,0) (A4) {};
                			
                		\draw[ ->,>=latex,thick] (A1) -- (A0);
                		\draw[ ->,>=latex,thick] (A2) -- (A1);
                		\draw[ ->,>=latex,thick] (A2) -- (A3);
            			\draw[ ->,>=latex,thick] (A3) -- (A4);
            		\end{tikzpicture}
                    \caption{\small Path graph.}
                    \label{fig:pathwith5}
                \end{figure}
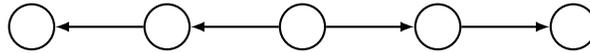
                
                Here, the minimum change in cost is achieved by an endpoint initiating a new channel to the player four edges away. We have
                $$\Delta \text{cost}(\text{add $1$ link}) = 1 - 2 \cdot b - 4 \cdot c.$$
                The path graph with five nodes is a Nash equilibrium for $1 \leq 2 \cdot b + 4 \cdot c$.
                \hfill \qed 
            \end{proof}
            
            \paththeorem*
              \begin{proof}
                To show that the path graph is never a Nash equilibrium for $n \geq 6$, we will show that at least one player in a path graph consisting of more than six players can always reduce her cost by changing strategy. 
                
                In a path graph with at least six players, at least one player $u$ has an outgoing edge to a player $v$ at least two steps from the end of the path on the opposite side of player $u$. This is illustrated in Figure~\ref{fig:path1} and we consider this to be strategy $s$. In this case it is always more beneficial for player $u$ to connect to player $w$ instead of player $v$. Let's refer to this strategy as strategy $\tilde{s}$ (Figure~\ref{fig:path2}). 
                
                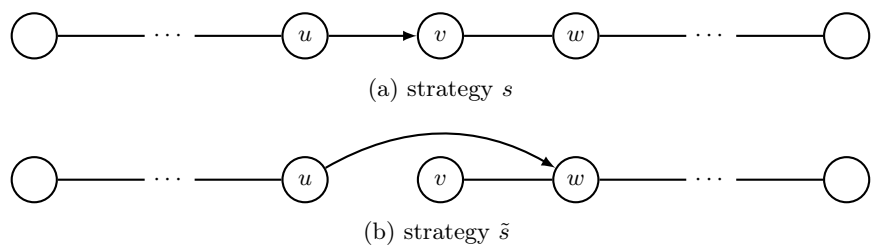
\begin{figure}[hbt!]
                    \centering
                    \begin{subfigure}[t]{\textwidth}
                        \centering
                        \begin{tikzpicture}[scale = 0.9]
                            %\tikzstyle{every node}=[font=\scriptsize]
                			\node[draw, circle, minimum size = 0.6cm,thick] at (0,0) (A0) {};
                			\node[] at (2,0) (A1) {$\cdots$};
                			\node[draw, circle, minimum size = 0.6cm,thick] at (4,0) (A2) {$u$};
                			\node[draw, circle, minimum size = 0.6cm,thick] at (6,0) (A3) {$v$};
                			\node[draw, circle, minimum size = 0.6cm,thick] at (8,0) (A4) {$w$};
                			\node[] at (10,0) (A5) {$\cdots$};
                			\node[draw, circle, minimum size = 0.6cm,thick] at (12,0) (A6) {};
                			
                			\draw[thick] (A0) -- (A1);
                			\draw[thick] (A1) -- (A2);
                			\draw[ ->,>=latex,thick] (A2) -- (A3);
                			\draw[thick] (A3) -- (A4);
                			\draw[thick] (A4) -- (A5);
                			\draw[thick] (A5) -- (A6);
                		\end{tikzpicture}
                        \caption{strategy $s$\label{fig:path1}}
                    \end{subfigure}\vspace{0.2cm}
                    
                    \begin{subfigure}[t]{\textwidth}
                        \centering
                    	\begin{tikzpicture}[scale = 0.9]
                    	    \node[draw, circle, minimum size = 0.6cm,thick] at (0,0) (A0) {};
                			\node[] at (2,0) (A1) {$\cdots$};
                			\node[draw, circle, minimum size = 0.6cm,thick] at (4,0) (A2) {$u$};
                			\node[draw, circle, minimum size = 0.6cm,thick] at (6,0) (A3) {$v$};
                			\node[draw, circle, minimum size = 0.6cm,thick] at (8,0) (A4) {$w$};
                			\node[] at (10,0) (A5) {$\cdots$};
                			\node[draw, circle, minimum size = 0.6cm,thick] at (12,0) (A6) {};
                			
                			\draw[thick] (A0) -- (A1);
                			\draw[thick] (A1) -- (A2);
                			\draw[thick] (A3) -- (A4);
                			\draw[thick] (A4) -- (A5);
                			\draw[thick] (A5) -- (A6);
                			\path[thick,->,>=latex] (A2) edge [bend left] node {} (A4);
                		\end{tikzpicture}	
                    	\caption{strategy $\tilde{s}$ \label{fig:path2}}
                    \end{subfigure}
                    \caption{\small Strategy deviation of player $1$.}
                \end{figure}
                
                The change in cost for this strategy is given as 
                $$ \Delta \text{cost}_u(s \text{ to } \tilde{s}) =  - c\cdot (m-2),$$
                where $m$ is the number of edges player $v$ is away from the endpoint on the opposite side $u$. Thus, the change in cost is negative and the path graph cannot be a Nash equilibrium for $n \geq 6$. \hfill \qed
            \end{proof}

        % \newpage
            
        \subsection{Circle Graph}
            \stara*
            \begin{proof}
                Adding a link to the only player one is not directly connected to, does not decrease a player's betweenness cost. It is not beneficial for a player to initiate an additional link, if the closeness cost reduction is not bigger than the link cost of one. Thus, a player does not add an additional edge for 
                $$c \leq 1.$$
                        
                A player's change in cost when removing a single link is given by 
                \begin{align*}
                    \Delta \text{cost}_u(\text{remove 1 link}) =& -1 + b + 2 \cdot c.
                \end{align*}
                Hence, a player cannot reduce her cost through the removal of a single link if $1 \leq b + 2 \cdot c$.
                    
                Additionally, a player with two outgoing links will never eliminate both without adding another link. Her cost would be infinite otherwise. Exchanging a single link with a new link to the player one was not previously connected to, never yields a negative change in cost and is therefore never a player's best response. 
                    
                Finally, the change in cost when removing two links and adding a new link to the player one was not previously connected to is given by  
                \begin{align*}
                    \Delta \text{cost}_u(\text{remove 2 \& add 1 link}) =& -1 + b + c,
                \end{align*}
                but this bound is more restrictive than the previous one, and there is no need for a player to have more than one outgoing edge.
                    
                Thus, the circle graph with $n=4$ is a Nash equilibrium for $c \leq 1 \leq b +2 \cdot c$.\hfill \qed 
            \end{proof}
                
            \starb*
            \begin{proof}
                The change in cost for the addition of links to the players, a player was not directly connected to previously, is given by 
                \begin{align*}
                    \Delta \text{cost}_u(\text{add $m$ links}) =& m - m \cdot b - m \cdot c,
                \end{align*}
                where $m \in \{1,2\}$. Thus, a player can reduce her cost by adding more links when $1 \leq b + c$.
                    
                If a player in the circle graph removes one outgoing link the change in cost is 
                \begin{align*}                    
                    \Delta \text{cost}_u(\text{remove 1 link}) =& -1 + 2 \cdot b+ 4 \cdot c.
                \end{align*}
                A player with an outgoing link benefits from the removal if $1 \geq 2 \cdot b + 4 \cdot c$. On the other hand, a player with two outgoing edges will never remove both links without adding a new link as the graph would become disconnected otherwise. Additionally, she never benefits more from exchanging links as the change in cost is non-negative. When replacing both her links by a new link, the change in cost is
                \begin{align*}
                    \Delta \text{cost}_u(\text{remove 2 \& add 1 link})=& -1 + 2 \cdot  b + 2 \cdot c. 
                \end{align*}
                However, this leads to a more restrictive bound than just removing one link and no player in a circle graph needs more than one outgoing link. 
                    
                We have shown that for $n=5$ the circle graph is a Nash equilibrium if and only if  
                $$ b + c \leq 1 \leq 2\cdot b +4\cdot c   .$$\hfill \qed 
            \end{proof}
            
            \circle*
              \begin{proof}
                We will show that any player with one outgoing edge in a circle graph with $n \geq N$ players, has an incentive to change strategy. Thus, the circle graph cannot be a Nash equilibrium. 
                
            	\begin{figure}[hbt!]
                    \centering
                    \begin{subfigure}[t]{0.48\textwidth}
                        \centering
                        \begin{tikzpicture}
                            \tikzstyle{every node}=[font=\small]
                			\def \n {6}
                			\def \radius {2.2cm}
                			\def \margin {14.5} % margin in angles, depends on the radius
                            \begin{scope}
            			        \clip (0,0)circle (\radius);
                                \draw[opacity =0.3] (-2.2,0)--(2.2,0);
                                \draw[opacity =0.3] (0,-2.2)--(0,2.2);
                                \node[align=center,anchor = center] at (1,0.35) {$1^{\text{st}}$ quadrant};
                                
                                \node[align=center,anchor = center] at (-1,0.35) {$4^{\text{th}}$ quadrant};
                               
                                \node[align=center,anchor = center] at (-1,-0.35) {$3^{\text{rd}}$ quadrant};
                                
                                \node[align=center,anchor = center] at (1,-0.35) {$2^{\text{nd}}$ quadrant};
                            \end{scope}

                			\node[draw, circle, minimum size = 1.1cm,thick, fill = white] at ({360/\n *0 +90}:\radius) (s0) {$0$};
                			\node[draw, circle, minimum size = 1.1cm,thick, fill = white] at ({360/\n *5 +90}:\radius) (s1) {$1$};
                			\node[draw, circle, minimum size = 1.1cm,thick, fill = white] at ({360/\n *3+ 90}:\radius) (sn2) {$\left \lfloor{\frac{n}{2}}\right \rfloor $};
                		    \node[draw, circle, minimum size = 1.1cm,thick, fill = white] at ({360/\n *1 + 90}:\radius) (sn1) {$n-1$};
                		    
                			\draw[>=latex, dotted,thick] ({360/\n * (5)-\margin +90}:\radius) arc ({360/\n * (5)-\margin +90}:{360/\n * (3)+\margin +90}:\radius);
                			\draw[>=latex, dotted,thick] ({360/\n * (3)-\margin +90}:\radius) arc ({360/\n * (3)-\margin +90}:{360/\n * (1)+\margin +90}:\radius);
                			\draw[>=latex,thick] ({360/\n * (1)-\margin +90}:\radius) arc ({360/\n * (1)-\margin +90}:{360/\n * (0)+\margin +90}:\radius);
                			\draw[->, >=latex,thick] ({360/\n * (0)-\margin +90}:\radius) arc ({360/\n * (0)-\margin +90}:{360/\n * (-1)+\margin +90}:\radius);
                		
                		\end{tikzpicture}
                        \caption{strategy $s$\label{fig:completecircle}}
                    \end{subfigure}
                    \hfill
                    \begin{subfigure}[t]{0.48\textwidth}
                        \centering
                        \begin{tikzpicture}
                        
                		    \tikzstyle{every node}=[font=\small]
                    		\def \n {6}
                    		\def \radius {2.2cm}
                    		\def \margin {14.5} % margin in angles, depends on the radius
                            
                            \begin{scope}
            			        \clip (0,0)circle (\radius);
                                \draw[opacity =0.3] (-2.2,0)--(2.2,0);
                                \node[align=center,anchor = center] at (1,0.35) {$1^{\text{st}}$ quadrant};
                               
                                \node[align=center,anchor = center] at (-1,0.35) {$4^{\text{th}}$ quadrant};
                                
                                \node[align=center,anchor = center] at (-1,-0.35) {$3^{\text{rd}}$ quadrant};
                                                        \node[align=center,anchor = center] at (1,-0.35) {$2^{\text{nd}}$ quadrant};
                            \end{scope}
                        
                    		\node[draw, circle, minimum size = 1.1cm,thick, fill =white] at ({360/\n *0 +90}:\radius) (s0) {$0$};
                    		\node[draw, circle, minimum size = 1.1cm,thick, fill =white] at ({360/\n *5 +90}:\radius) (s1) {$1$};
                    		\node[draw, circle, minimum size = 1.1cm,thick, fill =white] at ({360/\n *3+ 90}:\radius) (sn2) {$\left \lfloor{\frac{n}{2}}\right \rfloor $};
                    		\node[draw, circle, minimum size = 1.1cm,thick, fill =white] at ({360/\n *1 + 90}:\radius) (sn1) {$n-1$};
                    			
                    		\draw[>=latex, dotted,thick] ({360/\n * (5)-\margin +90}:\radius) arc ({360/\n * (5)-\margin +90}:{360/\n * (3)+\margin +90}:\radius);
                    		\draw[>=latex, dotted,thick] ({360/\n * (3)-\margin +90}:\radius) arc ({360/\n * (3)-\margin +90}:{360/\n * (1)+\margin +90}:\radius);
                    		\draw[>=latex,thick] ({360/\n * (1)-\margin +90}:\radius) arc ({360/\n * (1)-\margin +90}:{360/\n * (0)+\margin +90}:\radius);
                    			
                    		\draw[ ->,>=latex,thick] (s0) -- (sn2);
                    		
                    	\end{tikzpicture}
                    	\caption{strategy $\tilde{s}$ \label{fig:brokencircle}}
                    \end{subfigure}
                    \caption{\small Strategy change of player $0$.}
                \end{figure}
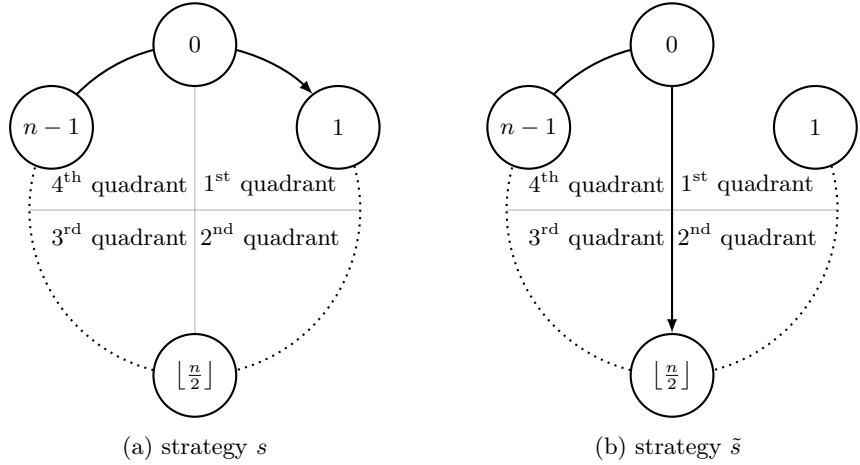
                
                Consider the circle graph in Figure~\ref{fig:completecircle}. Without loss of generality, assume that player $0$ has one outgoing edge to player $1$. As the equations for $0$'s betweenness and closeness differ for $n$ even or odd, we will use asymptotic notation throughout the following analysis. 
                
                In the circle graph (strategy $s$), the betweenness of $0$ is  
            	$$\text{betweenness}_0(s) = \dfrac{3}{4} \cdot n^2 + o\left(n^2\right) $$
            	and the $0$'s closeness is
            	$$\text{closeness}_0(s) = \dfrac{1}{4} \cdot n^2 + o\left(n^2\right).$$
                
            	Now, player $0$ removes the link to player $1$ and initiates a new link to player $\left \lfloor{\frac{n}{2}}\right \rfloor$, seen in Figure~\ref{fig:brokencircle}. We will refer to this strategy as $\tilde{s}$. The first part of $0$'s betweenness cost reduction comes from the shortest paths of players in the $1^{\text{st}}$ and $2^{\text{nd}}$ quadrant to the $4^{\text{th}}$ quadrant, as well as the other way around; the quadrants are as shown in Figure~\ref{fig:brokencircle}. These shortest paths go through the shortcut and subtract
            	$$2 \cdot \left( \frac{n}{2} + o(n)\right) \cdot \left( \frac{n}{4} + o(n) \right) = \frac{n^2}{4} + o\left(n^2\right), $$
            	from $0$'s betweenness cost. The second part stems from nodes in the $1^{\text{st}}$ and $3^{\text{rd}}$ quadrant using node $0$ as a gateway in the cycle. We have a further betweenness cost reduction of 
            	$$\frac{1}{4} \cdot \left( \frac{n}{2} \right)^2 + o\left(n^2\right)  =  \frac{n^2}{16}  + o\left(n^2\right)  .$$
            	Thus, the betweenness of $0$ with strategy $\tilde{s}$ is at most 
            	$$ \text{betweenness}_0(\tilde{s}) = n^2 -\frac{n^2}{4} - \frac{n^2}{16} + o\left(n^2\right)= \frac{11}{16} \cdot n^2 +o\left(n^2\right).$$

            	The closeness of player $0$ to players in the $3^{\text{rd}}$ and $4^{\text{th}}$ quadrant is 
            	$$ \dfrac{1}{4} \cdot \left( \dfrac{n}{2}\right)^2 + o\left(n^2\right) = \dfrac{n^2}{16} + o\left(n^2\right),$$
            	and to players in the $1^{\text{st}}$ and $2^{\text{nd}}$ quadrant $0$'s closeness is 
            	$$ \dfrac{1}{2} \cdot \left( \dfrac{n}{2}\right)^2 + o\left(n^2\right) =\dfrac{n^2}{8} +o\left(n^2\right) .$$
            	Therefore, $0$'s closeness is
            	$$ \text{closeness}_0(\tilde{s}) = \dfrac{n^2}{16} +\dfrac{n^2}{8} +o\left(n^2\right)  = \dfrac{3}{16} \cdot n^2 +o\left(n^2\right).$$
            	
                Player $0$'s change in cost is
            	\begin{align*}
                    \Delta \text{cost}_u(s \text{ to } \tilde{s}) =&  \left(\frac{11}{16}  n^2-\dfrac{3}{4}  n^2 + o\left(n^2\right) \right) \cdot b +\left( \dfrac{3}{16}  n^2 -\dfrac{1}{4}  n^2 + o\left(n^2\right) \right) \cdot c \\
                    =&  -\left(\frac{1}{16} n^2 + o\left(n^2\right) \right) (b+c). 
                \end{align*}
                As player $0$ would choose strategy $\tilde{s}$ over strategy $s$ for $\Delta \text{cost}_u(s \text{ to } \tilde{s}) <0$, there exists a $N>0$, such that for $n \geq N$ player the circle graph is never a Nash equilibrium. \hfill \qed 
            \end{proof}
            
            % \newpage
        \subsection{Price of Anarchy}
            \poaone*
            \begin{proof}
                The only Nash equilibrium for $c>1$ is the complete graph as stated by 
                Theorem~\ref{thm:thridtheorem}. As the social optimum for $c > \frac{1}{2} +b$ is also the complete graph (Theorem~\ref{thm:firsttheorem}), the price of anarchy is 
                $$\text{PoA} =1,$$
                for $ c > 1$ and $c > \frac{1}{2} +b$. \hfill \qed
            \end{proof}
            
            \poatwo*
            \begin{proof}
                For $c > 1$ and $b \leq c \leq \frac{1}{2} +b$ the only Nash equilibrium is the complete graph (Theorem~\ref{thm:thridtheorem}) and according to Theorem~\ref{thm:firsttheorem}, the social optimum is the star graph. Thus, the price of anarchy is given by 
            \begin{align*}
                \text{PoA} =& \dfrac{\text{cost}(\text{complete graph})}{\text{cost}(\text{star graph})}\\
                =& \dfrac{\left(\frac{1}{2}+ (n-2) \cdot b \right)(n-1)\cdot n}{(1-2(c-b)+(c-b)\cdot n+b\cdot (n-2)\cdot n)(n-1)}\\
                =& \dfrac{(\left(\frac{1}{2}+ (n-2) \cdot b \right)\cdot n}{1+ (c + b\cdot (n-1))(n-2)}.
            \end{align*} \hfill \qed
            \end{proof}
            
            \poathree*
            \begin{proof}
            For $1<c<b$ the only Nash equilibrium is the complete graph as stated in Theorem~\ref{thm:thridtheorem} and the social optimum is a path graph (Theorem~\ref{thm:firsttheorem}). The price of anarchy is given by 
            \begin{align*}
                \text{PoA} =& \dfrac{\text{cost}(\text{complete graph})}{\text{cost}(\text{path graph})}\\
                =& \dfrac{\left(\frac{1}{2}+ (n-2) \cdot b \right)(n-1)\cdot n}{\left( 1  + \left(\frac{2}{3}b + \frac{1}{3}c\right) \cdot n \cdot (n-2)\right) (n-1)}\\
                =& \dfrac{(\left(\frac{1}{2}+ (n-2) \cdot b \right)\cdot n}{1  + \left(\frac{2}{3}b + \frac{1}{3}c\right) \cdot n \cdot (n-2) }.
            \end{align*}\hfill \qed
        \end{proof}
        
        \poafour*
         \begin{proof}
            For $c>1$ and $c > \frac{1}{2}+b$, the price of anarchy is one and therefore it is also $\mathcal{O}(1)$. 
                
            We have that for $c>1$ and $b \leq c \leq \frac{1}{2} +b$,  
            $$\text{PoA} = \frac{\left(\frac{1}{2}+ (n-2) \cdot b \right)\cdot n}{1+ (c + b\cdot (n-1)) (n-2)} = \mathcal{O} \left(\frac{b\cdot n^2}{ b\cdot n^2} \right)  =  \mathcal{O}(1),$$
            and for $1<c<b$, 
            $$\text{PoA} = \frac{\left(\frac{1}{2}+ (n-2) \cdot b \right)\cdot n}{1  + \left(\frac{2}{3}b + \frac{1}{3}c\right) \cdot n \cdot (n-2)} =\mathcal{O} \left(\frac{b\cdot n^2}{ b\cdot n^2} \right)  =  \mathcal{O}(1).$$
            Thus, for $c>1$ we have $\text{PoA} =\mathcal{O}(1)$. \hfill \qed
        \end{proof}  
        
        \poafive*
        \begin{proof}
            For $c+b <\frac{1}{n^2}$, all Nash equilibria are trees. Unless the distance to a player is infinite, no player in the network has an incentive to build an edge. 
                    
            As both the maximum possible change in  $\text{betweenness}_u(s)$ and $\text{closeness}_u(s)$ for a node $u$ in a connected graph is less than $n^2$ and all Nash equilibria are connected, $\Delta \text{cost}_u (s) > - n ^2 \cdot c -  n ^2 \cdot b  +1$.
            We require $\Delta \text{cost}_u (s) \geq 0$ such that $u$ does not benefit from initiating an additional channel. Thus, for $c + b \leq \frac{1}{n^2}$ all Nash equilibria are spanning trees. 
                    
            For $c + b \leq \frac{1}{n^2}$  the social optimum is also a spanning tree, as it is either the star or path graph. It easily follows that for $c + b \leq \frac{1}{n^2}$ and all spanning trees $\text{cost}(s) = \Theta(n)$ and therefore the price of anarchy is $\mathcal{O}(1)$.
        \end{proof}
        
        \poasix*
        \begin{proof}
            The price of anarchy is 
            $$\text{PoA} = \mathcal{O} \left( \dfrac{\lvert E(G) \rvert + n^3 \cdot b + (c- b) \cdot  \sum  _{u \in [n]}  \sum  _{r \in [n]-u} \left( d_{G}(u,r) -1 \right) }{ n^3 \cdot b + n} \right).$$
            We can say that $d_{G}(u,r) < \Theta \left (\frac{2}{\sqrt{c+b}} \right)$, as player $u$ would connect to player $r$ otherwise. Player $u$ would become closer to half the nodes on the path otherwise and reduce her betweenness cost through the routing potential gained by the link addition. Therefore we have, 
            $$\text{PoA} = \mathcal{O} \left( \dfrac{\lvert E(G) \rvert + n^3 \cdot b + n^2 \frac{c-b}{\sqrt{b+c}} }{b \cdot n^3 + n} \right).$$
            It follows that
            \begin{align*}
                \mathcal{O} \left( \dfrac{n^3 \cdot b }{n^3 \cdot b   + n} \right) =& \mathcal{O}(1),
                &\text{and}& 
                &\mathcal{O} \left( \dfrac{n^2 \frac{c-b}{\sqrt{b+c}}}{n^3 \cdot b   + n} \right) =& \mathcal{O} \left( \dfrac{c-b}{n^2 \cdot b   + 1} \right) = \mathcal{O} \left( 1 \right),
            \end{align*}
            % $$\mathcal{O} \left( \dfrac{n^3 \cdot b }{n^3 \cdot b   + n} \right) = \mathcal{O}(1),$$ 
            % and 
            % $$\mathcal{O} \left( \dfrac{n^2 \frac{c-b}{\sqrt{b+c}}}{n^3 \cdot b   + n} \right) = \mathcal{O} \left( \dfrac{c-b}{n^2 \cdot b   + 1} \right) = \mathcal{O} \left( 1 \right),$$
            as  $c + b \geq \frac{1}{n^2}$ and $c<1$. Thus, it only remains to consider $\mathcal{O} \left( \frac{\lvert E(G) \rvert }{b \cdot n^3 + n} \right)$. 
                
            As $\lvert E(G) \rvert = \mathcal{O}(n^2)$ for any Nash equilibrium, we have $\text{PoA} = \mathcal{O}(n).$ \hfill \qed
        \end{proof}
        
    \subsection{Price of Stability}
    \pos*
            \begin{proof}
            As the price of stability is smaller than or equal to the price of anarchy, we can follow from Corollary \ref{cor:poa4}, that the price of stability is $\mathcal{O}(1)$ for $c>1$. Additionally, Theorem \ref{thm:poa} indicates that $\text{PoS} = \mathcal{O}(1)$ for $b + c< \frac{1}{n^2}$.
            \hfill \qed
        \end{proof}

\end{document}